\documentclass[11pt]{article}

\usepackage{amsfonts}
\usepackage{amsmath}
\usepackage{amssymb}
\usepackage{amsthm}
\usepackage{psfrag}

\usepackage{graphicx}

\theoremstyle{plain}
\newtheorem{theorem}{Theorem}

\newtheorem{definition}[theorem]{Definition}

\newtheorem{corollary}[theorem]{Corollary}

\theoremstyle{definition}
\newtheorem{example}[theorem]{Example}
\newtheorem{remark}[theorem]{Remark}

\numberwithin{equation}{section}
\numberwithin{theorem}{section}

\DeclareMathOperator{\Int}{int}

\DeclareMathOperator{\bd}{bd}

\DeclareMathOperator{\diag}{diag}

\newcommand{\bs}[1]{{\boldsymbol{#1}}}

\newcommand{\pdt}{\partial_t}
\newcommand{\R}{\mathbf{R}}

\newcommand{\D}{\,\mathrm{d}}

\begin{document}
\title{{Replicator equations and space}}

\author{Alexander S. Bratus$^{1,2}$, Vladimir P. Posvyanskii$^{2}$\,, Artem S. Novozhilov$^{{3},}$\footnote{Corresponding author: artem.novozhilov@ndsu.edu}\\[3mm]
\textit{\normalsize $^\textrm{\emph{1}}$Faculty of Computational Mathematics and Cybernetics,}\\[-1mm]
\textit{\normalsize Lomonosov Moscow State University, Moscow 119992, Russia}\\[2mm]
\textit{\normalsize $^\textrm{\emph{2}}$Applied Mathematics--1, Moscow State University of Railway Engineering,}\\[-1mm]\textit{\normalsize Moscow 127994, Russia}\\[2mm]
\textit{\normalsize $^\textrm{\emph{3}}$Department of Mathematics, North Dakota State University, Fargo, ND 58108, USA}}

\date{}

\maketitle

\begin{abstract}
A reaction--diffusion replicator equation is studied. A novel method to apply the principle of global regulation is used to write down the model with explicit spatial structure. Properties of stationary solutions together with their stability are analyzed analytically, and relationships between stability of the rest points of the non-distributed replicator equation and distributed system are shown. A numerical example is given to show that the spatial variable in this particular model promotes the system's permanence.

\end{abstract}

\section{Introduction}

The classical replicator equation \cite{hofbauer1998ega,hofbauer2003egd} models a wide array of different biological phenomena, including those in theoretical population genetics \cite{hofbauer1998ega}, evolutionary game theory \cite{hofbauer2003egd,nowak2006evolutionary}, or in theories of the origin of life \cite{eigen1989mcc}. In its general form, a replicator equations can be written as
\begin{equation}\label{eq0:1}
    \dot{w}_i=w_i\bigl(f_i(\bs w)-f^{l}(t)\bigr),\quad i=1,\ldots,n,
\end{equation}
for the vector $\bs w=(w_1,\ldots,w_n)$ of system variables. In the following we will speak of $\bs w$ as the vector of concentrations of macromolecules that interact with each other; however, different interpretations are possible. The interactions are modeled through the rate coefficients (fitnesses) $f_i(\bs w)$, which depend in general on the concentrations of other macromolecules. The expression $f^{l}(t)$ is necessary to keep the total concentration $\sum_{i=1}^nw_i$ constant. Very often $f_i(\bs w)=\sum_{j=1}^na_{ij}w_j$ for some real matrix $\bs A=(a_{ij})_{n\times n}$.

Model \eqref{eq0:1} is a system of ordinary differential equations, which implies that one assumes that there is no spatial structure in the studied system, or, in different words, the reactor that contains the macromolecules is so well stirred that any macromolecule has equal chance to interact with any other. For many systems such an assumption is a very crude approximation, therefore it is of significant interest to consider modifications of \eqref{eq0:1} that include explicit spatial structure.

There are different ways to add space to a mathematical model, and often the properties of the new model differ in an important way from the properties of the original mean-field model \cite{dieckmann2000}. A straightforward way to add space in ecological models by adding the Laplace operator to the right hand sides of the equations does not work for \eqref{eq0:1} because of the fact that there is an additional condition that $\sum_{i=1}^nw_i=\mbox{const}$ (see also \cite{novozhilov2012reaction} for a discussion). Various methods were utilized to overcome this obstacle, see, e.g., \cite{cressman1987density,cressman1997sad,ferriere2000ads,hutson1995sst,hutson1992travelling,Vickers1991,vickers1989spa,vickers1993spatial,weinberger1991ssa} and references therein. Our approach to tackle this particular problem is to use the principle of global regulation \cite{bratus2006ssc,bratus2009existence,bratus2011}. However, notwithstanding a number of interesting and new results concerning the spatially nonuniform stationary states, the equations that we studied in \cite{bratus2006ssc,bratus2009existence,bratus2011} behave similarly to the solutions of the non-distributed replicator equation \eqref{eq0:1} (this statement can be made precise, see the cited references). Therefore, for the purpose of the current work, we undertook a different approach to add space to the replicator equation \eqref{eq0:1}. We still keep the premise of the global regulation (see below), but the resulting equations have quite different properties, whose analytical and numerical analysis is the contents of the present manuscript.

The rest of the paper is organized as follows. In Section 2 we fix the notations and state the mathematical problem, which is in the center of our analysis. For the purpose of comparison we also present the reaction--diffusion replicator systems that we studied in \cite{bratus2006ssc,bratus2009existence,bratus2011}. Section 3 is devoted to the analysis of the stationary solutions of the corresponding replicator equations. We find the conditions under which the distributed system behaves like the mean-field model. In Section 4 stability of the stationary solutions is analyzed; again, we are able to prove that some particular knowledge on the stability of the stationary points in the non-distributed case can be used to infer the stability of the stationary state of the distributed system. From the biological point of view it is very important to guarantee that none of the macromolecules go extinct with the time, this condition is formalized mathematically using the notions of \textit{persistence} and \textit{permanence} (e.g., \cite{cantrell2003spatial}). In Section 5 we obtain a sufficient condition for our system to be persistent. Although a great deal of analysis can be accomplished analytically (Sections 2--5), the replicator equation that we study actually possesses the property that even if in the non-distributed system some of the species go extinct, the distributed reaction-diffusion replicator equation supports the existence of all macromolecules; in Section 6 we give an example of such behavior, basing our particular model on the \textit{in vitro} experiments of RNA self-replicating molecules \cite{Vaidya2012}. Finally, Section 7 is devoted to concluding remarks and comments.

\section{Problem statement and notations}

Let $\Omega$ be a bounded domain in $\R^m$ with a piecewise smooth boundary $\Gamma$, and $m$ denotes the dimension of the problem, we consider only $m=1,2,$ or $3$. Without loss of generality we assume that $|\Omega|=1$, i.e., the volume of $\Omega$ is equal to 1. Denote $N_k(\bs x,t)$ the number of macromolecules of the $k$-th type, $k=1,\ldots,n,$ per volume unit at the time moment $t$ at the point $\bs x\in\Omega\subset \R^m$. We postulate that the relative rate of change of $N_k(\bs x,t)$ at the point $\bs x\in\Omega$ is governed by the following law
\begin{equation}\label{eq1:1}
    \frac{\pdt N_k(\bs x,t)}{N_k(\bs x,t)}=\bigl(\bs{A N}(\bs x,t)\bigr)_k+d_k\Delta N_k(\bs x,t),\quad k=1,\ldots,n,
\end{equation}
where $\pdt N_k(\bs x,t)=\frac{\partial N_k}{\partial t}(\bs x,t),$ $\bs N(\bs x,t)=(N_1(\bs x,t),\ldots, N_n(\bs x,t))$, $\bs A=(a_{ks})_{n\times n}$ is an $n\times n$ real matrix, $$\bigl(\bs{A N}(\bs x,t)\bigr)_k=\sum_{s=1}^na_{ks} N_s(\bs x,t),$$ $\Delta$ is the Laplace operator, in Cartesian coordinates $\Delta =\sum_{i=1}^m\frac{\partial^2}{\partial x_i^2}\,,\,m=1,2,3$, $d_k>0$ are the numbers that characterize the influence of the uniform diffusion on the rate of change of the densities $N_k(\bs x,t),\,k=1,\ldots,n$.

A possible interpretation of \eqref{eq1:1} is that we consider a porous medium diffusion equation
$$
\phi_k \pdt N_k=f_k(\bs N)+\Delta N_k,\quad k=1,\ldots,n,
$$
for which the porosity $\phi_k$ depends on the local concentrations $N_k(\bs x,t)$: $\phi_k(N_k(\bs x,t))=N_k^{-1}(\bs x,t)$, which means that inverse in the number of particles reduces the space available (for the diffusion equation in a porous medium see, e.g., \cite{barenblatt1989theory,knabner2003numerical}).

The initial conditions are
\begin{equation}\label{eq1:2}
    N_k(\bs x,t)=\Psi_k(\bs x),\quad k=1,\ldots,n,
\end{equation}
and the boundary conditions are
\begin{equation}\label{eq1:3}
    \left.\frac{\partial N_k(\bs x,t)}{\partial \nu}\right|_{\bs x\in \Gamma}=0,\quad k=1,\ldots,n,
\end{equation}
where $\nu$ is the outward normal to the boundary $\Gamma$ of $\Omega$. Condition \eqref{eq1:3} describes the zero flux of the macromolecules through boundary $\Gamma$.

System \eqref{eq1:1}--\eqref{eq1:3} defines a \textit{selection system} \cite{karev2010}. It is usually more convenient to replace such system with the corresponding \textit{replicator equation}, which describes the change of frequencies (see, e.g., \cite{novozhilov2012reaction}).

Let
$$
\Sigma(t)=\sum_{i=1}^n\int_\Omega N_i(\bs x,t)\D\bs x,
$$
and assume for the following that $\Sigma(t)>0$ for any $t\geq 0$. Then the corresponding frequencies of macromolecules are defined as
$$
v_k(\bs x,t)=\frac{N_k(\bs x,t)}{\Sigma(t)}\,,\quad k=1,\ldots,n.
$$
By construction we have
\begin{equation}\label{eq1:4}
    \sum_{k=1}^n\int_{\Omega} v_k(\bs x,t)\D \bs x=1.
\end{equation}
Direct calculations lead to
$$
\pdt v_k(\bs x,t)=\Sigma(t)v_k(\bs x,t)\Bigl(\bigl(\bs A\bs v(\bs x,t)\bigr)_k-f^{sp}(t)+d_k\Delta v_k(\bs x,t)\Bigr),\quad k=1,\ldots,n,
$$
where
$$
f^{sp}(t)=\int_\Omega \Bigl(\langle \bs{Av}(\bs x,t),\bs v(\bs x,t)\rangle+\sum_{k=1}^n d_kv_k(\bs x,t)\Delta v_k(\bs x,t)\Bigr)\D \bs x.
$$
Hereinafter $\langle\cdot,\cdot\rangle$ denotes the standard inner product in $\R^n$, $\bs v(\bs x,t)=(v_1(\bs x,t),\ldots,v_n(\bs x,t))$. Note that from \eqref{eq1:3} it follows that for the frequencies $v_k(\bs x,t)$ the boundary conditions
\begin{equation}\label{eq1:5}
    \left.\frac{\partial v_k(\bs x,t)}{\partial \nu}\right|_{\bs x\in \Gamma}=0,\quad k=1,\ldots,n
\end{equation}
hold. Therefore, using Green's identity, the expression for $f^{sp}(t)$ can be rewritten as
\begin{equation}\label{eq1:6}
f^{sp}(t)=\int_\Omega \Bigl(\langle \bs{Av}(\bs x,t),\bs v(\bs x,t)\rangle-\sum_{k=1}^n d_k \sum_{i=1}^m \left[\frac{\partial v_k(\bs x,t)}{\partial x_i}\right]^2\Bigr)\D \bs x.
\end{equation}
The last term in \eqref{eq1:6} can be rewritten in the form $\sum_{k=1}^nd_k\langle\nabla v_k,\nabla v_k\rangle=\sum_{k=1}^n d_k\|\nabla v_k\|^2$. Note that expression \eqref{eq1:6} for any $t\geq 0$ is a functional defined on the set of vector-functions $\bs v(\bs x,t)$.

Letting $t=\int_{0}^\tau \Sigma(\varsigma)\D\varsigma$, we finally obtain the system
\begin{equation}\label{eq1:7}
   \partial_{\tau}v_k(\bs x,\tau)=v_k(\bs x,\tau)\Bigl(\bigl(\bs A\bs v(\bs x,\tau)\bigr)_k-f^{sp}(\tau)+d_k\Delta v_k(\bs x,\tau)\Bigr),\quad k=1,\ldots,n,
\end{equation}
with the initial conditions
\begin{equation}\label{eq1:8}
    v_k(\bs x,0)=\varphi_k(\bs x),\quad k=1,\ldots,n,
\end{equation}
which follow from \eqref{eq1:2}, and boundary conditions \eqref{eq1:5}.

Note that from \eqref{eq1:7} and equality \eqref{eq1:6}, taking into account \eqref{eq1:5}, we have that
$$
\frac{\partial}{\partial \tau}\Bigl(\sum_{k=1}^n\int_\Omega v_k(\bs x,\tau)\D \bs x\Bigr)=0,
$$
which corresponds to \eqref{eq1:4}.

In the following we will call the functional $f^{sp}(t)$ \textit{the mean fitness} of the population of macromolecules, whereas the quantity $\bigl(\bs{Av}(\bs x,t)\bigr)_k$ will be referred to as the \textit{fitness} of the $k$-th macromolecule at the point $\bs x\in\Omega$ at the time moment $\tau$. From now on we will also use the variable $t$ instead of $\tau$, keeping in mind that this is a rescaled time.

System \eqref{eq1:4}--\eqref{eq1:8} will be called \textit{the reaction--diffusion replicator equation with the global regulation of the second kind} as opposed to the reaction--diffusion replicator equation with the global regulation of the first kind (see \cite{novozhilov2012reaction} for a concise review on the reaction--diffusion replicator systems and \cite{bratus2012diffusive,bnp2010,bratus2006ssc,bratus2009existence,bratus2011} for an in-depth analysis of such systems). We recall that in the cited papers dynamics and the limit behavior of the replicator systems of the form
\begin{equation}\label{eq1:9}
    \pdt v_k(\bs x,t)=v_k(\bs x,t)\Bigl(\bigl(\bs{Av}(\bs x,t)\bigr)_k-f_1^{sp}(t)\Bigr)+d_k\Delta v_k(\bs x,t),\quad k=1,\ldots,n,
\end{equation}
was studied. In \eqref{eq1:9} the mean fitness $f_1^{sp}(t)$ is given by
\begin{equation}\label{eq1:10}
    f_1^{sp}(t)=\int_\Omega \langle\bs{Av}(\bs x,t),\bs v(\bs x,t)\rangle\D\bs x,
\end{equation}
and $v_k(\bs x,t)$ are nonnegative functions satisfying \eqref{eq1:4}, \eqref{eq1:5}, and \eqref{eq1:8}.

Therefore, systems \eqref{eq1:4}--\eqref{eq1:8} and \eqref{eq1:9}--\eqref{eq1:10} differ both by the form of the equations and by the expressions for the mean population fitness, coinciding in the limit $d_k\to 0$. The mean fitness $f_1^{sp}(t)$ of the system \eqref{eq1:9}, \eqref{eq1:10} does not depend on the spatially non-uniform distribution and coincide in the form with the usual mean fitness for the non-distributed replicator equation \cite{hofbauer1998ega,hofbauer2003egd}. At the same time, the mean fitness $f^{sp}(t)$ of \eqref{eq1:4}--\eqref{eq1:8} includes the dependence on the square of the expression that characterizes the rate of change of the form of the spatially non-uniform distribution of the densities. This is an important feature of the reaction--diffusion replicator equation with the global regulation of the second kind. Eventually, both of the systems \eqref{eq1:4}--\eqref{eq1:8} and \eqref{eq1:9}--\eqref{eq1:10} represent possible generalizations of the classical replicator equation \eqref{eq0:1} for the case of explicit spatial structure under different principles of global regulation.

In the following we assume that the functions $v_k(\bs x,t),\,\bs x\in\Omega,\,t\geq 0,\,k=1,\ldots,n$ are differentiable with respect to $t$, and, together with their derivatives with respect to $t$, belong to the Sobolev space $W^{1,2}(\Omega)$ if $m=1$, or to $W^{2,2}(\Omega)$ if $m=2,3$, as functions of the variable $\bs x\in\Omega$ for each fixed $t$. Here $W^{r,2}(\Omega)$ are the usual Sobolev spaces such that their elements belong to $L^2(\Omega)$ together with all the (weak) derivatives up to the order $r$. We note that the embedding theorems imply that the elements of $W^{r,2}(\Omega),\,r=1,2$ coincide with continuous functions on $\Omega$ almost everywhere (e.g., \cite{evans_2010}).

Denote $\Omega_t=\Omega\times[0,\infty)$ and consider the space of functions $B(\Omega_t)$ with the norm
$$
\|y\|_{B(\Omega_t)}=\max_{t\geq 0}\bigl\{\|y(\bs x,t)\|_{W^{r,2}}+\|\pdt y(\bs x,t)\|_{W^{r,2}} \bigr\},\quad r=1,2.
$$
Let $S_n(\Omega_t)$ denote the set of functions from $B(\Omega_t)$ for which \eqref{eq1:4} holds. The set $S_n(\Omega_t)$ is \textit{an integral simplex} in the space $B(\Omega_t)$. Together with $S_n(\Omega_t)$ also consider the set $S_n(\Omega)$ of the vector-functions $\bs u(\bs x)=(u_1(\bs x),\ldots,u_n(\bs x))$ such that $u_k(\bs x)\in W^{r,2}(\Omega)$ ($r=1,2$) for which
\begin{equation}\label{eq1:11}
    \sum_{k=1}^n\int_\Omega u_k(\bs x)\D \bs x=1
\end{equation}
holds. The set $S_n(\Omega)$ is \textit{an integral simplex} in the space $W^{r,2}(\Omega),\,(r=1,2)$. We consider weak solutions to the system \eqref{eq1:4}--\eqref{eq1:8}, i.e., such solutions $\bs v(\bs x,t)\in S_n(\Omega_t)$ for which the following integral equalities hold
\begin{align*}
\int_0^\infty\int_{\Omega}\pdt v_k(\bs x,t)\eta(\bs x,t)\D \bs x\D t&=\int_0^\infty\int_{\Omega}v_k(\bs x,t)\bigl((\bs{Av}(\bs x,t))_k-f^{sp}(t)\bigr)-\\
&-d_k \int_0^\infty\int_{\Omega}\sum_{i=1}^nd_i\Bigl(\sum_{l=1}^m\left[\frac{\partial v_i}{\partial x_l}\right]^2+\sum_{l=1}^m\frac{\partial v_i}{\partial x_l}\frac{\partial \eta}{\partial x_l}\Bigr)\D\bs x\D t
\end{align*}
for any function $\eta=\eta(\bs x,t)$ on compact support that is differentiable with respect to $t$, and for each $t\geq 0$ belongs to $W^{r,2}(\Omega),\,r=1,2$.

Together with the problem \eqref{eq1:4}--\eqref{eq1:8} consider the classical replicator equation (e.g., \cite{hofbauer1998ega})
\begin{equation}\label{eq1:12}
   \dot{w}_k(t)=w_k(t)\Bigl(\bigl(\bs{Aw}(t)\bigr)_k-f^{l}(t)\Bigr),\quad k=1,\ldots,n,
\end{equation}
where $f^l(t)=\langle \bs{Aw}(t),\bs w(t)\rangle$. The system \eqref{eq1:12} is defined on the simplex $S_n$ of smooth nonnegative functions $\bs w(t)=(w_1(t),\ldots,w_n(t))$ such that
$$
\sum_{k=1}^nw_k(t)=1
$$
for any $t$.

We will also need the definitions of the boundary and interior sets of the (integral) simplex $S_n(\Omega_t)$ (or $S_n(\Omega)$, or $S_n$).

\begin{definition}\label{def1:1} The boundary set $\bd S_n(\Omega_t)$ of $S_n(\Omega_t)$ is the set of vector-functions $\bs v(\bs x,t)\in S_n(\Omega_t)$ such that for some indexes $k\in K$ in the subset $K\subset\{1,\ldots,n\}$ one has
$$
\overline{v}_k(t)=0,\quad k\in K,\,t\geq 0,
$$
where
$$
\overline{v}_k(t)=\int_\Omega v_k(\bs x,t)\D \bs x.
$$

The interior set $\Int S_n(\Omega_t)$ is the set of functions $\bs v(\bs x,t)\in S_n(\Omega_t)$ such that
$$
\overline{v}_k(t)>0,\quad k=1,\ldots,n,
$$
for any $t\geq 0$.
\end{definition}
Note that due to \eqref{eq1:4} we have for the elements of $\bd S_n(\Omega_t)$
$$
\sum_{k\notin K}\int_\Omega v_k(\bs x,t)\D \bs x=1.
$$

Analogously, for the system \eqref{eq1:12} and the standard simplex $S_n$ its boundary and interior sets are defined, respectively, as the set which has at least one coordinate $w_k(t)=0$ and the set for which all the coordinates $w_k(t)>0$ for $k=1,\ldots,n$. Note that these sets are invariant for \eqref{eq1:12}.

\begin{remark} Since $v_k(\bs x,t)\in W^{r,2}(\Omega),\,r=1,2$ for any fixed $t\geq 0$, this implies that $v_k(\bs x,t)$ coincide with continuous functions almost everywhere. Therefore, $\overline{v}_k(t)=0$ implies that $v_k(\bs x,t)=0$ almost everywhere in $\Omega$.
\end{remark}

\begin{remark}\label{rm1:3}
For any element $\bs v(\bs x,t)\in\bd S_n(\Omega_t)$ an element $\bs w(t)\in \bd S_n$ can be identified. Indeed, we can always put $\bs w(t)=\overline{\bs v}(t)$.
\end{remark}

\section{Stationary solutions to the distributed replicator equation}
The stationary solutions to the problem \eqref{eq1:4}--\eqref{eq1:8} are determined by the following time independent system of equations
\begin{equation}\label{eq2:1a}
    u_k(\bs x)\Bigl(\bigl(\bs{Au}(\bs x)\bigr)_k-\overline{f}^{sp}+d_k\Delta u_k(\bs x)\Bigr)=0,\quad k=1,\ldots,n,
\end{equation}
with the boundary conditions
\begin{equation}\label{eq2:1b}
    \left.\frac{\partial u_k(\bs x)}{\partial \nu}\right|_{\bs x\in \Gamma}=0,\quad k=1,\ldots,n,
\end{equation}
and
\begin{equation}\label{eq2:1c}
    \overline{f}^{sp}=\int_\Omega \bigl(\langle\bs{Au}(\bs x),\bs{u}(\bs x)\rangle-\sum_{k=1}^{n}d_k\|\nabla u_k(\bs x)\|^2\bigr)\D \bs x.
\end{equation}
Solutions to \eqref{eq2:1a}--\eqref{eq2:1c} will be sought both in the set $\Int S_n(\Omega)$ and in the set $\bd S_n(\Omega)$. Together with the solutions to \eqref{eq2:1a}--\eqref{eq2:1c}, consider the stationary points of \eqref{eq1:12}, which are given as the solutions to
\begin{equation}\label{eq2:2}
    w_k\bigl((\bs{Aw})_k-\langle\bs{Aw},\bs{w}\rangle\bigr)=0,\quad \bs w\in S_n,\quad k=1,\ldots, n.
\end{equation}

Consider an auxiliary eigenvalue problem
\begin{equation}\label{eq2:3}
    \Delta \psi(\bs{x})+\lambda \psi(\bs{x})=0,\quad \bs{x}\in\Omega,\quad \partial_\nu\psi|_{\bs{x}\in\Gamma}=0.
\end{equation}
The eigenfunction system of \eqref{eq2:3} is given by $\psi_0(\bs{x})=1,\,\{\psi_i(\bs{x})\}_{i=1}^\infty$ and forms a complete system in the Sobolev space $W^{r,2}(\Omega),\,r=1,2$ (e.g., \cite{mikhlin1964vmm}), additionally
\begin{equation}\label{eq2:4}
    \int_\Omega \psi_i(\bs{x})\psi_j(\bs{x})\D\bs{x}=\delta_{ij},
\end{equation}
where $\delta_{ij}$ is the Kronecker symbol. The corresponding eigenvalues satisfy the condition
$$
0=\lambda_0<\lambda_1\leq \lambda_2\leq\ldots\leq\lambda_i\leq\ldots,\qquad \lim_{i\to\infty}\lambda_i=+\infty.
$$

It it convenient to introduce the following definition.

\begin{definition} We shall call the diffusion coefficient $d_k$ of the system \eqref{eq2:1a} $\mu$-resonant if there exists an eigenvalue $\lambda_{s_k}$ of \eqref{eq2:3} such that
\begin{equation}\label{eq2:5}
    d_k=\frac{\mu}{\lambda_{s_k}}\,.
\end{equation}
\end{definition}
By definition $\mu$ has to a be a positive constant. From the definition it follows that if
\begin{equation}\label{eq2:6}
    d_{\min}>\frac{\mu}{\lambda_1}\,,\quad d_{\min}=\min\{d_1,\ldots,d_n\}
\end{equation}
for a given $\mu$, then all the diffusion coefficients are not $\mu$-resonant.

\begin{theorem}\label{th2:1} Let $\bs A$ have at least one real eigenvalue and let $\mu$ be the maximal eigenvalue of $\bs A$. Assume also that system \eqref{eq2:2} has an isolated solution $\bs w\in\Int S_n$. If condition \eqref{eq2:6} holds then all the stationary solutions $\bs u(\bs x)\in \Int S_n(\Omega)$ of the distributed system \eqref{eq1:4}--\eqref{eq1:8} are spatially uniform and coincide with the interior rest point $\bs w\in \Int S_n$ of \eqref{eq1:12}.
\end{theorem}
\begin{proof}
Let $\bs u(\bs x)\in\Int S_n(\Omega)$ be a solution to \eqref{eq2:1a}--\eqref{eq2:1c}. Then
\begin{equation}\label{eq2:7}
    \bigl(\bs{Au}(\bs x)\bigr)_k-\overline{f}^{sp}+d_k\Delta u_k(\bs x)=0,\quad \partial_\nu u_k=0,\quad k=1,\ldots,n.
\end{equation}
Let us look for a solution to \eqref{eq2:7} in the form of a series with the basis $\psi_0(\bs{x})=1,\,\{\psi_i(\bs{x})\}_{i=0}^\infty$ of solutions to \eqref{eq2:3}:
\begin{equation}\label{eq2:8}
    u_k(\bs x)=\overline{u}_k+U_k(\bs x),\quad U_k(\bs x)=\sum_{s=1}^\infty c_s^k\psi_s(\bs x),
\end{equation}
where
$$
\overline{u}_k=\int_\Omega \psi_0(\bs x)u_k(\bs x)\D \bs x=\int_\Omega u_k(\bs x)\D \bs x.
$$
Putting \eqref{eq2:8} into \eqref{eq2:7}, integrating through $\Omega$, and taking into account \eqref{eq2:4}, we find
\begin{equation}\label{eq2:9}
    (\bs{A\overline{u}})_1=\ldots=(\bs{A\overline{u}})_n=\overline{f}^{sp}.
\end{equation}
Then from \eqref{eq2:8} and \eqref{eq2:7} it follows that for each $k$
\begin{equation}\label{eq2:10}
    \bigl(\bs{AU}(\bs x)\bigr)_k=-d_k\Delta U_k(\bs x),\quad \bs{U}(\bs x)=\bigl(U_1(\bs x),\ldots, U_n(\bs x)\bigr),
\end{equation}
which implies
$$
d_k\sum_{s=1}^\infty c_s^k\lambda_s\psi_s(\bs x)=\sum_{j=1}^na_{kj}\sum_{s=1}^\infty c_s^j\psi_s(\bs x).
$$
Taking the inner product of the last equality with $\psi_s(\bs x)$ in $L^2(\Omega)$ and using \eqref{eq2:4}, we obtain that the coefficients $c_s^k$ have to be solutions of the following linear systems
\begin{equation}\label{eq2:11}
    (\bs A-d_k\lambda_s \bs I)\bs c^s=0,\quad \bs c^s=(c_s^1,\ldots,c_s^n)^\top,\quad s=1,2,\ldots
\end{equation}
where $\bs I$ is the identity matrix. If the assumptions of the theorem hold, then systems \eqref{eq2:11} have only trivial solutions $\bs c^s=0$, which implies that $u_k(\bs x)=\overline{u}_k$ for all $k$ and $\overline{f}^{sp}=\langle\bs{A\overline{u}},\bs{\overline{u}}\rangle$, therefore the solution to \eqref{eq2:9} coincides with the solution to \eqref{eq2:2}.
\end{proof}

The results of Theorem \ref{th2:1} can be generalized for the case when solutions to \eqref{eq2:1a}--\eqref{eq2:1c} are taken from $\bd S_n(\Omega)$. Let $\bs u(\bs x)\in\bd S_n(\Omega)$. Then there exists set $K\subset\{1,\ldots,n\}$ from Definition \ref{def1:1}. Denote $\bs A_K$ the matrix which is obtained from $\bs A$ by deleting the rows and columns with indexes from $K$.

\begin{corollary}
Let $\bs A_K$ have at leat one real eigenvalue and let $\mu$ be the maximal eigenvalue of $\bs A_K$. Also assume that there exists an isolated solution $\bs w\in\bd S_n$ to \eqref{eq2:2} with $\bs A=\bs A_K$. If condition \eqref{eq2:6} on the set of $k\notin K$ holds, then all the solutions $\bs u(\bs x)\in\bd S_n(\Omega)$ are spatially uniform and coincide with the solutions $\bs w\in \bd S_n$ to \eqref{eq1:12}.
\end{corollary}
The proof follows the steps of the proof of Theorem \ref{th2:1}.

\begin{corollary}\label{corr2:5}Let $\bs A$ have at least one real eigenvalue and let $\mu$ be the maximal eigenvalue of $\bs A$. Assume also that system \eqref{eq2:2} has an isolated solution $\bs w\in\Int S_n$. If problem \eqref{eq2:1a}--\eqref{eq2:1c} possesses $\mu$-resonant diffusion coefficients, then it has infinitely many spatially nonuniform solutions, whose mean integral values coincide with the solution to \eqref{eq2:2}.
\end{corollary}
\begin{proof}
In this case systems \eqref{eq2:11} have infinitely many solutions. If each of equalities in \eqref{eq2:10} is multiplied by $U_k(\bs x)$ and integrated over $\Omega$, then we obtain
\begin{equation}\label{eq2:12}
\int_\Omega \langle \bs{AU}(\bs x),\bs{U}(\bs x)\rangle\D \bs x=\sum_{k=1}^n d_k\int_\Omega \langle \nabla U_k(\bs x),\nabla U_k(\bs x)\rangle\D \bs x.
\end{equation}
From the representation of $\overline{f}^{sp}$ in \eqref{eq2:1c} and series \eqref{eq2:8} we have
\begin{equation}\label{eq2:13}
    \langle \bs{A\overline{u}},\bs{\overline{u}}\rangle=\overline{f}^{sp}.
\end{equation}
The equality \eqref{eq2:13} together with \eqref{eq2:9} imply that $\bs{\overline{u}}=\bs{w}$, where $\bs w$ solves \eqref{eq2:2}.
\end{proof}


\begin{example} Consider stationary solutions when $\bs A$ is a circulant matrix:
$$
\bs A=\begin{bmatrix}
        a_1 & a_2 & \ldots & a_{n-1} & a_n \\
        a_n & a_1 & \ldots & a_{n-2} & a_{n-1} \\
        \ldots &  &  &  &  \\
        a_2 & a_3 & \ldots & a_n & a_1 \\
      \end{bmatrix}.
$$
If $\mu=\sum_{i=1}^na_i>0$ then, according to Corollary \ref{corr2:5}, there are infinitely many spatially nonhomogeneous solutions to \eqref{eq2:1a}--\eqref{eq2:1c} if, e.g.,
$$
d_{\max}=\frac{\mu}{\lambda_1}\,,\quad d_{\max}=\max\{d_1,\ldots,d_n\}.
$$
If $\mu<0$, then the number of stationary solutions is finite and they coincide with the solutions to \eqref{eq2:2} corresponding to the replicator equation \eqref{eq1:12}.
\end{example}

Before closing this section, we would like to stress an interesting feature of the problem \eqref{eq2:1a}--\eqref{eq2:1c}: It is possible to find such $\bs u(\bs x)\in \Int S_n(\Omega)$ satisfying the boundary condition \eqref{eq2:1b}, that on some domain $\Omega_k\subset\Omega$ the function $u_k(\bs x)\equiv 0$ when $\bs x\in\Omega_k$, whereas on $\Omega\setminus\Omega_k$ $u_k(\bs x)$ satisfies \eqref{eq2:1a} and \eqref{eq2:1c}. We illustrate this assertion with an example.

\begin{example}Consider an autocatalytic system on $\Omega=(0,1)$. This means that the matrix $\bs A$ is diagonal, $\bs A=\diag(a_1,\ldots,a_n)$, and the stationary solutions solve
\begin{align*}
u_k(x)&\Bigl(d_k\frac{d^2u_k}{dx^2}(x)+a_ku_k(x)-\overline{f}^{sp}\Bigr)=0,\quad a_k>0,\quad k=1,\ldots,n,\\
u'_k(0)&=u'_k(1)=0,\quad \overline{f}^{sp}=\sum_{k=1}^n\int_0^1\Bigl(a_ku_k^2(x)-d_k\bigl(u'_k(x)\bigr)^2\Bigr)\D x.
\end{align*}
Assume that
$$\sqrt{\frac{a_1}{d_1}}=m\pi$$
for some positive integer $m\geq 1$, whereas for the rest of the parameters $a_k/d_k<\pi$ for $k=2,\ldots,n$. Then solutions have the form
$$
u_1(x)=c_1\cos m\pi x+\frac{\overline{f}^{sp}}{a_1}\,,\quad u_k(x)=\frac{\overline{f}^{sp}}{a_k}\,,\quad k=2,\ldots,n,
$$
and since $\bs u(x)\in S_n(0,1)$ then the following condition
$$
\overline{f}^{sp}=\frac{1}{\sum_{k=1}^na_k^{-1}}
$$
should hold. This can be checked directly, since
$$
\overline{f}^{sp}=\int_0^1\Bigl(a_1u_1^2(x)-d_1 \bigl(u'_1(x)\bigr)^2+\sum_{k=2}^na_ku_k^2(x)\Bigr)\D x=\bigl(\overline{f}^{sp}\bigr)^2\sum_{k=1}^n\frac{1}{a_k}\,,
$$
which yields the required equality. To guarantee that $u_1(x)$ is nonnegative, it is enough to require $|c_1|\leq \overline{f}^{sp}/a_1$. Since $c_1$ is arbitrary, we found infinitely many stationary solutions (cf. Corollary \ref{corr2:5}). Apart from this set, as it can be directly verified, the following choices for $u_1(x)$ are also solutions:
\begin{align*}
u_1(x)&=\begin{cases}\frac{\overline{f}^{sp}}{a_1}(1+\cos m\pi x), &0<x<\frac{1}{m}\\
0,&\frac{1}{m}\leq x< 1
\end{cases}\\
u_1(x)&=\begin{cases}0, &0<x\leq \frac{m-1}{m}\\
 \frac{\overline{f}^{sp}}{a_1}(1+(-1)^m\cos m\pi x), &\frac{m-1}{m}<x<1
\end{cases}\\
u_1(x)&=\begin{cases}0, &0<x\leq \frac{m-2}{2m}\\
 \frac{\overline{f}^{sp}}{a_1}(1+(-1)^m\cos m\pi x), &\frac{m-2}{2m}<x<\frac{m+2}{2m}\\
 0,&\frac{m+2}{2m}\leq x<1
\end{cases}
\end{align*}
The list of examples can be continued. Moreover, similar examples can be constructed in case when $\Omega$ is a rectangular area in $\R^2$ or in $\R^3$. The key conditions for such solutions to appear is the existence of $\mu$-resonant diffusion coefficients.
\end{example}

\section{Stability of the stationary solutions}

\begin{theorem}\label{th3:1}Let $\mu$ be the maximal real part of the eigenvalues of matrix $\bs A$.  Assume also that system \eqref{eq2:2} has an isolated solution $\bs w\in\Int S_n$. If condition \eqref{eq2:6} holds then the asymptotic stability (or instability) of the interior rest point $\bs w\in\Int S_n$ of the replicator equation \eqref{eq1:12} implies asymptotic stability (or instability) of the interior stationary solution to \eqref{eq1:4}--\eqref{eq1:8}.
\end{theorem}
\begin{proof} Theorem \ref{th2:1} yields that the interior stationary point coincide with the interior stationary point of \eqref{eq1:12}. Fix an $\varepsilon>0$ and look for a solution to \eqref{eq1:4}--\eqref{eq1:8} in the form
\begin{equation}\label{eq3:1}
    v_k(\bs x,t)=w_k+c_k^0(t)+\sum_{s=1}^\infty c_s^k\psi_s(\bs x),\quad k=1,\ldots,n,
\end{equation}
where $\psi_k(\bs x)$ are the eigenfunctions of the problem \eqref{eq2:3}, assuming that the initial conditions $v_k(\bs x,0)=\varphi_k(\bs x)$ satisfy
\begin{equation}\label{eq3:1a}
    \|\varphi_k(\bs x)-w_k\|_{L^2(\Omega)}<\delta,\quad k=1,\ldots,n.
\end{equation}

Plugging \eqref{eq3:1} into \eqref{eq1:7}, integrating over $\Omega$ and keeping only linear terms with respect to $c_k^0(t)$, we obtain the following system of linear equations:
\begin{equation}\label{eq3:2}
    \frac{dc_k^0(t)}{dt}=w_k\left[\bigl(\bs{Ac}^0\bigr)_k-\langle\bs A^\top\bs w,\bs c^0(t)\rangle-\langle\bs{Aw},\bs c^0(t)\rangle\right]+c_k^0\left[\bigl(\bs{Aw}\bigr)_k-\langle\bs{Aw},\bs w\rangle\right],
\end{equation}
where $k=1,\ldots,n$. By virtue of
$$
1=\sum_{k=1}^n\int_\Omega v_k(\bs x,t)\D \bs x=\sum_{k=1}^nw_k+\sum_{k=1}^nc_k^0(t)
$$
and
$$
\sum_{k=1}^n w_k=1,
$$
we have
\begin{equation}\label{eq3:3}
    \sum_{k=1}^nc_k^0(t)=0.
\end{equation}
Therefore, from \eqref{eq2:2} it follows that
$$
\langle\bs{Aw},\bs c^0(t)\rangle=\sum_{k=1}^n c_k^0(t)\bigl(\bs{Aw}\bigr)_k=f^l\sum_{k=1}^nc_k^0(t)=0,
$$
and
$$
\langle\bs{Aw},\bs w\rangle=\bigl(\bs{Aw}\bigr)_k.
$$
System \eqref{eq3:2} now reads
\begin{equation}\label{eq3:4}
   \frac{dc_k^0(t)}{dt}=w_k\left[\bigl(\bs{Ac}^0\bigr)_k-\langle\bs A^\top\bs w,\bs c^0(t)\rangle\right].
\end{equation}
The Jacobi matrix of system \eqref{eq3:4} coincides with the Jacobi matrix of \eqref{eq1:12} evaluated at the interior stationary point $\bs w\in\Int S_n$ given by \eqref{eq2:2}. Therefore, if $\bs w$ is asymptotically stable (unstable), then the trivial solution to \eqref{eq3:4} is also asymptotically stable (unstable).

Now we plug \eqref{eq3:1} into \eqref{eq1:7}, multiply consecutively by $\psi_i(\bs x)$ and integrate over $\Omega$; keeping only linear terms with respect to $\bs c^s(t)=(c_s^1(t),\ldots,c_s^n(t))$, we obtain the linear systems of equations of the form
\begin{equation}\label{eq3:5}
    \frac{d\bs c^s(t)}{dt}=\bs W(\bs A-\lambda_s \bs D)\bs c^s(t),\quad s=1,2,\ldots,
\end{equation}
where $\bs W=\diag(w_1,\ldots,w_n)$ and $\bs D=\diag(d_1,\ldots,d_k)$. By the assumptions of the theorem, the trivial solution to \eqref{eq3:5} is asymptotically stable. To prove this fact, it is sufficient to consider a Lyapunov function $V_s(t)=\langle\bs W^{-1}\bs c^s(t),\bs c^s(t)\rangle$ and use the properties of the spectrum of the problem \eqref{eq2:3}.

Putting together the last two observations we obtain that choosing $\delta$ small enough in \eqref{eq3:1a} yields that
$$
\|v_k(\bs x,t)-w_k\|_{L^2(\Omega)}<\varepsilon,\quad t\geq 0,\quad k=1,\ldots,n.
$$
This implies that the asymptotic stability (instability) of the interior solution to \eqref{eq1:7} coincides with asymptotic stability (instability) of the interior solution to \eqref{eq2:2}.
\end{proof}

\begin{corollary}
Let $\mu_K$ be the maximal real part of the eigenvalues of the matrix $\bs A_K$, which is obtained from $\bs A$ by removing rows and columns with the indexes from the set $K\subset\{1,\ldots,n\}$ and let $\bs w\in\bd S_n$ be a rest point of \eqref{eq1:12} such that $w_i=0$ if $i\in K$. Then if the condition \eqref{eq2:6} is satisfied with $\mu_K$ instead of $\mu$, then the asymptotic stability (instability) of $\bs w$ implies the asymptotic stability (instability) of these solutions as stationary points of \eqref{eq1:7}.
\end{corollary}

For the following we introduce
\begin{definition}
A stationary solution $\bs u(\bs x)\in S_n(\Omega)$ to system \eqref{eq1:4}--\eqref{eq1:8} is called stable in the sense of the mean integral value, if for any $\varepsilon>0$ there exists a $\delta>0$ such that for any initial conditions \eqref{eq1:8} that satisfy
\begin{equation}\label{eq3:6}
    |\overline{\psi}_k-\overline{u}_k|<\delta,\quad k=1,\ldots, n,
\end{equation}
where
$$
\overline{\psi}_k=\int_\Omega \psi_k(\bs x)\D \bs x,\quad \overline{u}_k=\int_\Omega u_k(\bs x)\D \bs x,
$$
it follows that
\begin{equation}\label{eq3:7}
    |\overline{v}_k(t)-\overline{u}_k|<\varepsilon,\quad k=1,\ldots,n,\quad t>0,
\end{equation}
where
$$
\overline{v}_k(t)=\int_\Omega v_k(\bs x,t)\D \bs x,
$$
and $v_k(\bs x,t)$ are the solutions to \eqref{eq1:7} with the initial conditions $\psi_k(\bs x)$.
\end{definition}

The stability in the sense of the mean integral value follows from the usual Lyapunov stability, whereas the opposite is not true (see \cite{bratus2011}). If in \eqref{eq3:7} we addionally have that $|\overline{v}_k(t)-\overline{u}_k|\to 0$ when $t\to\infty$ then we shall call $\bs u(\bs x)$ \textit{asymptotically stable in the sense of the mean integral value}.

\begin{corollary}
Let the reaction--diffusion replicator equation \eqref{eq1:4}--\eqref{eq1:8} have $\mu$-resonant diffusion coefficients, where $\mu$ is the maximal real part of the eigenvalues of $\bs A$, then the asymptotic stability (instability) of $\bs w\in S_n$ of the problem \eqref{eq1:12} implies asymptotic stability (instability) of $\bs u(\bs x)$ of the problem \eqref{eq1:4}--\eqref{eq1:8} in the sense of the mean integral value.
\end{corollary}
\begin{proof}
Let us look for the solution to \eqref{eq1:4}--\eqref{eq1:8} that satisfies condition \eqref{eq3:6} in the form
$$
v_k(\bs x,t)=\overline{u}_k+c_k^0(t)+\sum_{s=1}^\infty c_s^k\psi_s(\bs x),\quad k=1,\ldots,n,
$$
such that $|\overline{v}_k(t)-\overline{u}_k|=|c_k^0(0)|<\delta,$ for some $\delta>0$. Corollary \ref{corr2:5} yields that the mean integral values $\overline{u}_k$ of the spatially nonuniform solutions $u_k(\bs x)$ are the stationary points of the replicator equation \eqref{eq1:12}. Therefore, $\overline{u}_k=w_k$, where $\bs w$ solves \eqref{eq2:2}. Note that \eqref{eq3:3} holds, and for $c_k^0(t)$ we obtain the linear approximation \eqref{eq3:4}, therefore, if $\bs w$ is asymptotically stable, then $c_k^0(t)$ tend to $0$ and $\bs u(\bs x)$ is asymptotically stable in the sense of the mean integral value.
\end{proof}

\section{Replicator dynamics}

Prior to stating the main theorem here, we give a definition of persistence and recall the specific form of Poincar\'{e}'s inequality that we use.

\begin{definition} The replicator equation defined on the integral simplex $S_n(\Omega_t)$ is said to be persistent if the initial conditions \eqref{eq1:8} $\overline{\varphi}_k>0$ for $k=1,\ldots,n$ imply
\begin{equation}\label{eq4:1}
    \liminf_{t\to\infty} \overline{v}_k(t)>0,\quad k=1,\ldots,n.
\end{equation}
where
$$
\overline{v}_k(t)=\int_\Omega v_k(\bs x,t)\D \bs x.
$$
\end{definition}
For $v_k(\bs x,t),\,k=1,\ldots,n$ it means that they are not zero almost everywhere in $\Omega$.

We will use  Poincar\'{e}'s inequality in the following form. Let $g(\bs x)\in W^{r,2},\,r=1,2$. There exist nonnegative constants $c_1$ and $c_2$, which depend on the geometry of $\Omega$ and do not depend on $g(\bs x)$, such that
$$
\int_{\Omega}g^2(\bs x)\D \bs x\leq c_1\int_\Omega \|\nabla g(\bs x)\|^2\D\bs x+c_2\Bigl(\int_{\Omega} g(\bs x)\D \bs x\Bigr)^2.
$$
In the particular case when $\int_{\Omega} g(\bs x)\D \bs x=0$ we have
\begin{equation}\label{eq4:1a}
    \int_{\Omega}g^2(\bs x)\D \bs x\leq c_1\int_\Omega \|\nabla g(\bs x)\|^2\D\bs x,
\end{equation}
and $c_1=\lambda_1^{-1}$, where $\lambda_1$ is the minimal nonzero eigenvalue of the problem \eqref{eq2:3} (see \cite{rektorys1980variational}).
\begin{theorem}\label{th4:1}
Assume that
\begin{equation}\label{eq4:2}
    \lambda_1 d_{\min}\geq \mu,\quad d_{\min}=\{d_1,\ldots,d_n\},
\end{equation}
and $\mu$ is the spectral radius of $\bs A$. If there exists a vector $\bs p\in\Int S_n$ for which
\begin{equation}\label{eq4:3}
    \langle \bs{Aw},\bs p\rangle-\langle\bs{Aw},\bs w\rangle>0
\end{equation}
for all points $\bs w\in\bd S_n$, then system \eqref{eq1:4}--\eqref{eq1:8} is persistent.
\end{theorem}
\begin{proof}
Consider a functional, depending on the variable $t$, on the set $S_n(\Omega_t)$,
$$
F(\bs v)=F\bigl(\bs v(\bs x,t)\bigr)=\exp\Bigl(\sum_{k=1}^n p_k\overline{\log v_k}(\bs x,t)\Bigr),\quad \bs v(\bs x,t)\in S_n(\Omega_t),
$$
where $\bs p\in \Int S_n$,
$$
\overline{\log v_k}(\bs x,t)=\int_\Omega \log v_k(\bs x,t)\D \bs x.
$$
Note that $F(\bs v)>0$, if $\bs v\in\Int S_n(\Omega_t)$.

Consider a sequence of vector-functions $\bs v^s(\bs x,t),\,s=1,2,\ldots$, that converges to some element $\bs v(\bs x,t)\in\bd S_n(\Omega_t),\,t>0$. Hence for some indexes $k\in K$ we have
$$
\overline{v}_k^s(t)=\int_\Omega v_k^s(\bs x,t)\D \bs x\to0,\quad s\to \infty,\,k\in K.
$$
By Jensen's inequality
$$
\overline{\log v_k}(\bs x,t)\leq \log \overline{v}_k(t),\quad k=1,\ldots,n.
$$
From the convergence of $\overline{v}_k^s(t)$ to zero and the last inequality we have
$$
F\bigl(\bs v^s(\bs x,t)\bigr)\leq \exp \Bigl(\sum_{k=1}^n p_k\log \overline{v}_k^s(t)\Bigr)=\prod_{k=1}^n \Bigl(\overline{v}_k^s(t)\Bigr)^{p_k}\to 0,
$$
for $\bs p\in\Int S_n$. Therefore $F(\bs v)$ is equal to zero on the set $\bd S_n(\Omega)$.

We also note that
$$
\frac{d}{dt}F(\bs v)=\dot{F}(\bs v)=F(v)\sum_{k=1}^n p_k\dot{\overline{\log v_k}}(\bs x,t),
$$
where the dot denotes the derivative with respect to time.

Rewrite the reaction--diffusion replicator equation \eqref{eq1:7} as
$$
\frac{\partial}{\partial t}\log v_k(\bs x,t)=\bigl(\bs{Av}(\bs x,t)\bigr)_k-f^{sp}(t)+d_k\Delta v_k (\bs x,t), \quad k=1,\ldots,n,\quad \bs v\in\Int S_n(\Omega_t).
$$
On integrating with respect to $\bs x\in\Omega$ we have
$$
\dot{\overline{\log v_k}}(\bs x,t)=\bigl(\bs{A\overline{v}}(t)\bigr)_k-f^{sp}(t),\quad k=1,\ldots,n.
$$
Therefore,
\begin{equation}\label{eq4:4}
    \dot{F}(\bs v)=F(\bs v)\bigl(\langle \bs{A\overline{v}}(t),\bs p\rangle-f^{sp}(t)\bigr).
\end{equation}

Consider the representation
\begin{equation}\label{eq4:5}
    v_k(\bs x,t)=\overline{v}_k(t)+V_k(\bs x,t),\quad k=1,\ldots,n,
\end{equation}
where
$$
\overline{v}_k(t)=\int_\Omega v_k(\bs x,t)\D\bs x,\quad V_k(\bs x,t)=\sum_{s=1}^\infty c_s^k(t)\psi_s(\bs x).
$$
From \eqref{eq4:5} it follows that
\begin{equation}\label{eq4:6}
f^{sp}(t)=\langle\bs{A\overline{v}}(t),\bs{\overline{v}}(t)\rangle+\int_\Omega \langle\bs{AV}(\bs x,t),\bs{V}(\bs x,t)\rangle\D \bs x-\sum_{k=1}^n\int_\Omega d_k\|\nabla V_k(\bs x,t)\|^2\D\bs x.
\end{equation}
Since
$$
\int_\Omega V_k(\bs x,t)\D\bs x=0,\quad k=1,\ldots,n,
$$
by Poincar\'{e}'s inequality \eqref{eq4:1a} we have
$$
\lambda_1 \int_\Omega \bigl(V_k(\bs x,t)\bigr)^2\D \bs x\leq \int_\Omega \|\nabla V_k(\bs x,t)\|^2\D\bs x,
$$
where $\lambda_1$ is the first nonzero eigenvalue of the problem \eqref{eq2:3}.

If $\mu$ is the spectral radius of $\bs A$, then
$$
|\langle\bs{AV}(\bs x,t),\bs{V}(\bs x,t)\rangle|\leq \mu\sum_{k=1}^n |V_k(\bs x,t)|^2.
$$
As a result, from \eqref{eq4:6} we have
$$
f^{sp}(t)\leq \langle\bs{A\overline{v}}(t),\bs{\overline{v}}(t)\rangle+\sum_{k=1}^n(\mu-d_k\lambda_1)\int_\Omega \bigl(V_k(\bs x,t)\bigr)^2\D \bs x.
$$
If \eqref{eq4:2} holds then
\begin{equation}\label{eq4:7}
    f^{sp}(t)\leq \langle\bs{A\overline{v}}(t),\bs{\overline{v}}(t)\rangle,
\end{equation}
therefore \eqref{eq4:4} yields
\begin{equation}\label{eq4:8}
    \dot{F}(\bs v)\geq F(\bs v)\bigl(\langle \bs{A\overline{v}}(t),\bs p\rangle -\langle \bs{A\overline{v}}(t),\bs{\overline{v}}(t)\rangle\bigr).
\end{equation}
Let us use Remark \ref{rm1:3} and identify $\bs v(\bs x,t)\in\bd S_n(\Omega_t)$ with a $\bs w(t)\in\bd S_n$ such that $\bs w(t)=\overline{\bs{v}}(t)$

From inequality \eqref{eq4:8} it follows that
$$
F(\bs v)\geq C\exp\Bigl\{\int_0^t\bigl(\langle \bs A\bs{\overline{v}}(t),\bs p\rangle-\langle \bs A\bs{\overline{v}}(t),\bs{\overline{v}}(t)\rangle\bigr)\D t\Bigr\},
$$
where $C=F(\bs v)|_{t=0}$.

Assume that there exists a $t>0$ such that $\bs v(\bs x,t)\in\bd S_n(\Omega_t)$. Then $F(\bs v)=0$. On the other hand, using Remark \ref{rm1:3}, we can identify any  $\bs v(\bs x,t)\in\bd S_n(\Omega_t)$ with $\bs w(t)\in \bd S_n$, then, using \eqref{eq4:3}, we must have
$$
F(\bs v)\geq C>0,
$$
which proves that the system is persistent.
\end{proof}

\begin{remark} To validate condition \eqref{eq4:3} is an independent algebraic problem.

Here is one possible approach. Assume that the vector $\bigl(\bs{A}^{\top}\bigr)^{-1}\bs{1}$ is positive. Here $\bs 1=(1,\ldots,1)^{\top}\in\R^n$. Consider $\bs p\in\Int S_n$
$$
\bs p=\frac{\bigl(\bs{A}^{\top}\bigr)^{-1}\bs{1}}{\langle\bigl(\bs{A}^{\top}\bigr)^{-1}\bs{1},\bs 1\rangle}\,.
$$
For any $\bs w\in\bd S_n$, one has
$$
\langle \bs{Aw},\bs p\rangle=\frac{1}{\langle\bigl(\bs{A}^{\top}\bigr)^{-1}\bs{1},\bs 1\rangle}\,.
$$
On the other hand
$$
\langle \bs{Aw},\bs w\rangle\leq \mu \langle \bs{w},\bs w\rangle,
$$
where $\mu$ is the spectral radius of $\bs A$. Since $\langle \bs{w},\bs w\rangle\leq 1$ for any $\bs w\in\bd S_n$, then the inequality
$$
\mu<\frac{1}{\langle\bigl(\bs{A}^{\top}\bigr)^{-1}\bs{1},\bs 1\rangle}
$$
yields the condition \eqref{eq4:3}.

To illustrate this approach, consider a very simple example.
\begin{example}
Consider the following replicator system with the global regulation of the second kind
\begin{equation}\label{eq4a:1}
    \begin{split}
      \pdt u_1&=u_1(\beta u_2+k_2u_2-f^{sp}(t)+d_1\Delta u_1),  \\
\pdt u_2&=u_2(k_2u_1-f^{sp}(t)+d_2\Delta u_2),
    \end{split}
\end{equation}
$\bs x\in\Omega$ and $\partial_{\nu}u_i=0,\,\bs x\in\Gamma$. Using the approach outlined above, we obtain
$$
p_1=\frac{k_1}{k_1+k_2-\beta}\,,\quad p_2=\frac{k_2-\beta}{k_1+k_2-\beta}\,.
$$
The condition \eqref{eq4:3} takes the form
$$
(k_1w_2+\beta w_1)p_1+k_2w_2p_2-(k_1+k_2)w_1w_2-\beta w_1^2>0.
$$
This is obviously true for $w_1=0,\,w_2=1$. For the case $w_1=1,\,w_2=0$ we have
$$
\frac{\beta k_1-(k_1+k_2)\beta+\beta^2}{k_1+k_2-\beta}>\frac{\beta^2-k_2\beta}{k_1+k_2-\beta}\,.
$$
The last expression will be positive if we require $k_1>\beta>k_2$. Finally, consider, e.g., the square $\Omega=(0,1)\times(0,1)$. In this area the condition \eqref{eq4:2} yields
$$
d_{\min}\geq \frac{\beta+\sqrt{\beta^2+4k_1k_2}}{8\pi^2}\,.
$$
\end{example}

The estimate in Theorem \ref{th4:1} gives only sufficient condition, as it can be seen, for instance, from the hypercycle replicator equation with the matrix
$$
\bs A=\begin{bmatrix}
        0 & 0 & \ldots & 0 & a_1 \\
        a_2 & 0 & \ldots & 0 & 0 \\
        0 & a_3 & 0 & \ldots & 0 \\
        \vdots & \ddots & \ddots &  & \vdots \\
        0 & \ldots & 0 & a_n & 0 \\
      \end{bmatrix}.
$$
It is well known that the hypercyclic system is not only persistent, it is permanent \cite{hofbauer1998ega}, i.e., the variables are separated from zero by a positive constant. Condition \eqref{eq4:3} holds only for the short hypercycles ($n=2,3,4$). Indeed, for $n=2$ we have $\langle \bs{Aw},\bs{w}\rangle=(a_1+a_2)w_1w_2=0,\,\bs w\in S_n$. For $n=3$ \eqref{eq4:3} holds if we choose
$$
p_i=\frac{a_i^{-1}}{R_3},\quad i=1,2,3,\quad R_3=\sum_{j=1}^3\frac{1}{a_j}\,,\quad \mbox{and }R_3\max\{a_1,a_2,a_3\}\leq 4\,,\quad i=1,2,3.
$$
Condition \eqref{eq4:2} yields here
$$
d_{\min}\geq \frac{(a_1a_2a_3)^{1/3}}{\lambda_1}\,.
$$

For the $n=4$ \eqref{eq4:3} will hold for a similar choice of $\bs p$ only for $a_1=a_2=a_3=a_4$.
\end{remark}

\begin{remark} In \cite{bratus2011} we suggested a generalization of the classical notions of the Nash equilibrium and evolutionary stable state for the case of the distributed reaction--diffusion replicator equation with the global regulation of the first kind. Similar definitions and results can be stated for the case of the global regulation of the second kind.

In particular, it can be proved that if a stationary solution $\bs u(\bs x)\in S_n(\Omega)$ to the distributed system \eqref{eq1:4}--\eqref{eq1:8} is Lyapunov stable then
$$
\int_\Omega \langle \bs v(\bs x,t),\bs{Au}(\bs x)\rangle\D \bs x\leq \overline{f}^{sp}(\bs u)
$$
for any $\bs v(\bs x,t)\in S_n(\Omega_t)$. Here $\overline{f}^{sp}(\bs u)$ is defined by \eqref{eq2:1c}.

Moreover, if $\bs u(\bs x)\in \Int S_n(\Omega)$ is a stationary solution to \eqref{eq1:4}--\eqref{eq1:8} and
$$
\int_\Omega \langle \overline{\bs u},\bs{Av}(\bs x,t)\rangle\D \bs x> {f}^{sp}\bigl(\bs v(\bs x,t)\bigr),
$$
for any $\bs v(\bs x,t)$ from a small neighborhood of $\bs u(\bs x)$, then $\bs u(\bs x)$ is asymptotically stable in the sense of the mean integral value. The proofs follow the lines of Theorem 2 and Theorem 3 in~\cite{bratus2011}.
\end{remark}

\section{Numerical analysis of a particular replicator system}
In this section we present an example which possesses an interesting and important feature: The non-distributed replicator equation is shown to be non-permanent (for the chosen parameter values three members of the catalytic network go extinct), whereas the distributed replicator equation with the global regulation of the second kind is permanent (for the same parameter values all six members of the catalytic network are subject to time dependent oscillations).

Consider a replicator system with the matrix
\begin{equation}\label{eq5:1}
    \bs A=\left[
\begin{array}{cccccc}
 0 & 0 & \alpha  & 0 & 0 & \gamma  \\
 \alpha  & 0 & 0 & 0 & \gamma  & 0 \\
 0 & \alpha  & 0 & \gamma  & 0 & 0 \\
 \gamma  & 0 & 0 & \beta  & 0 & 0 \\
 0 & 0 & \gamma  & 0 & \beta  & 0 \\
 0 & \gamma  & 0 & 0 & 0 & \beta  \\
\end{array}
\right].
\end{equation}
The catalytic network which corresponds to the interaction matrix \eqref{eq5:1} presented in Fig.~\ref{fig:1}. \begin{figure}[!t]
\centering
\includegraphics[width=0.4\textwidth]{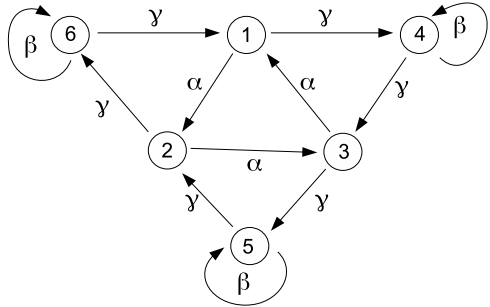}
\caption{A catalytic network of macromolecules. There are six macromolecules. The arrows show the catalytic activity of the molecules. The coefficients are the corresponding rate constants. This network is inspired by the catalytic network of self-replicating RNA molecules, which was shown to be capable of sustained replication \cite{novozhilov2012reaction}}\label{fig:1}
\end{figure}This particular cooperative network, which contains two catalytic cycles, is based on the in vitro network of RNA molecules, which was shown to be capable of sustained self-replication~\cite{Vaidya2012}.

Our task is to compare the behavior of solutions of three different analytical approaches to model this network: Classical local replicator equation \eqref{eq1:12}, reaction--diffusion replicator equation with the global regulation of type one \eqref{eq1:9} and reaction-diffusion replicator equation with the global regulation of type two \eqref{eq1:7}.

Let the parameters take the values
$$
\alpha=1.75,\quad \beta=0.7,\quad \gamma=2.0.
$$
For these parameter values it can be shown that there are fifteen rest points  of \eqref{eq1:12} belonging to $S_n$, including one isolated rest point in $\Int S_n$. However, this interior rest point is unstable. Moreover, numerical experiments show that there are several stable rest points such that three coordinates stay positive whereas other three species go extinct. In general, the conclusion is that for the taken parameter values the system is not permanent and cannot guarantee survival of all the molecules. We do not give illustrations here because the distributed model with the global regulation of the first kind shows very similar behavior (in full accordance with the theoretical analysis in \cite{bratus2011}).

Now consider the replicator equation with the global regulation of the first type \eqref{eq1:9} on $\Omega=(0,1)$ with Neumann's boundary conditions. The initial conditions for all the subsequent calculations are shown in Fig. \ref{fig:2}. The details of the numerical scheme are discussed elsewhere \cite{bratus2006ssc}.
\begin{figure}[!b]
\centering
\includegraphics[width=0.4\textwidth]{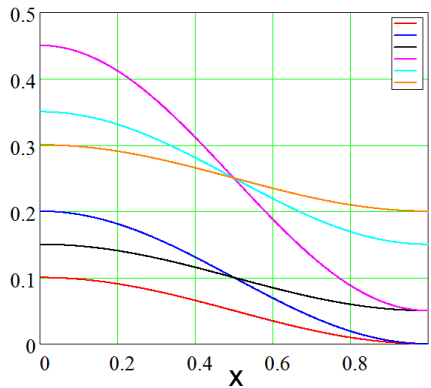}
\caption{Initial conditions for solving problems \eqref{eq1:9} and \eqref{eq1:7} on $\Omega=(0,1)$}\label{fig:2}
\end{figure}
We take two different vectors of the diffusion coefficients: $\bs d_1=(0.4,0.5,0.4,0.5,0.4,0.5)\,,\bs d_2=(0.04,0.05,0.04,0.05,0.04,0.05)$. As it was proved in \cite{bratus2011}, for larger diffusion coefficients the system is actually becomes spatially homogeneous for sufficiently large $t$. For example, in Fig. \ref{fig:3} it is sufficient to take $t=60$; by this time moments the distributions of the species are spatially uniform.  The right panel in Fig.~\ref{fig:3} shows the time evolution for the mean values of the variables
$$
\overline{v}_i(t)=\int_{\Omega}v_i(\bs x,t)\D x,\quad i=1,\ldots,6.
$$
It can be seen that after the initial transitory period, the solutions actually attracted to the spatially homogeneous (left panel) stationary state, which corresponds exactly to the asymptotically stable rest point of the non-distributed replicator equation \eqref{eq1:12}.
\begin{figure}[!t]
\includegraphics[width=0.45\textwidth]{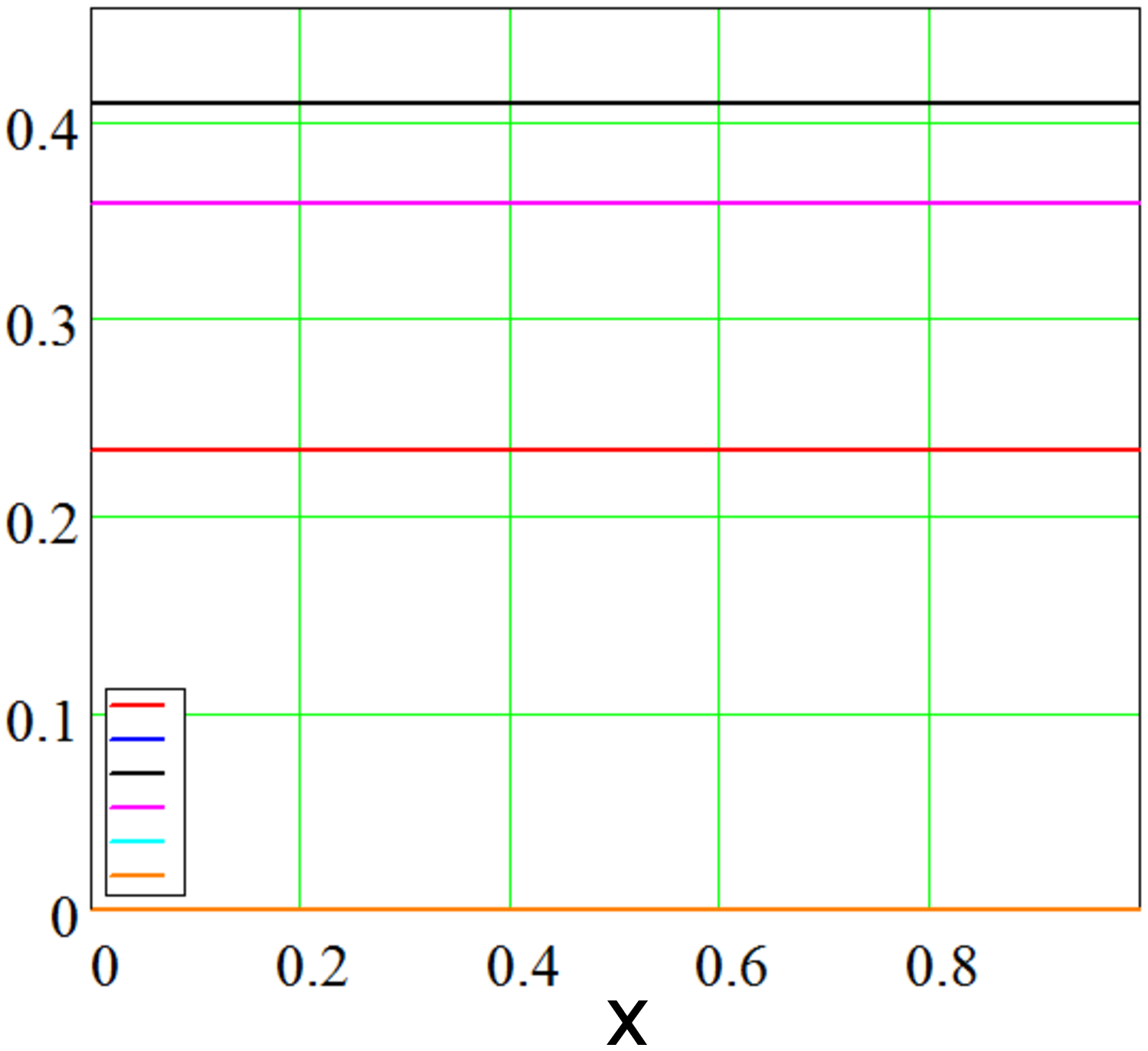}\hfill
\includegraphics[width=0.495\textwidth]{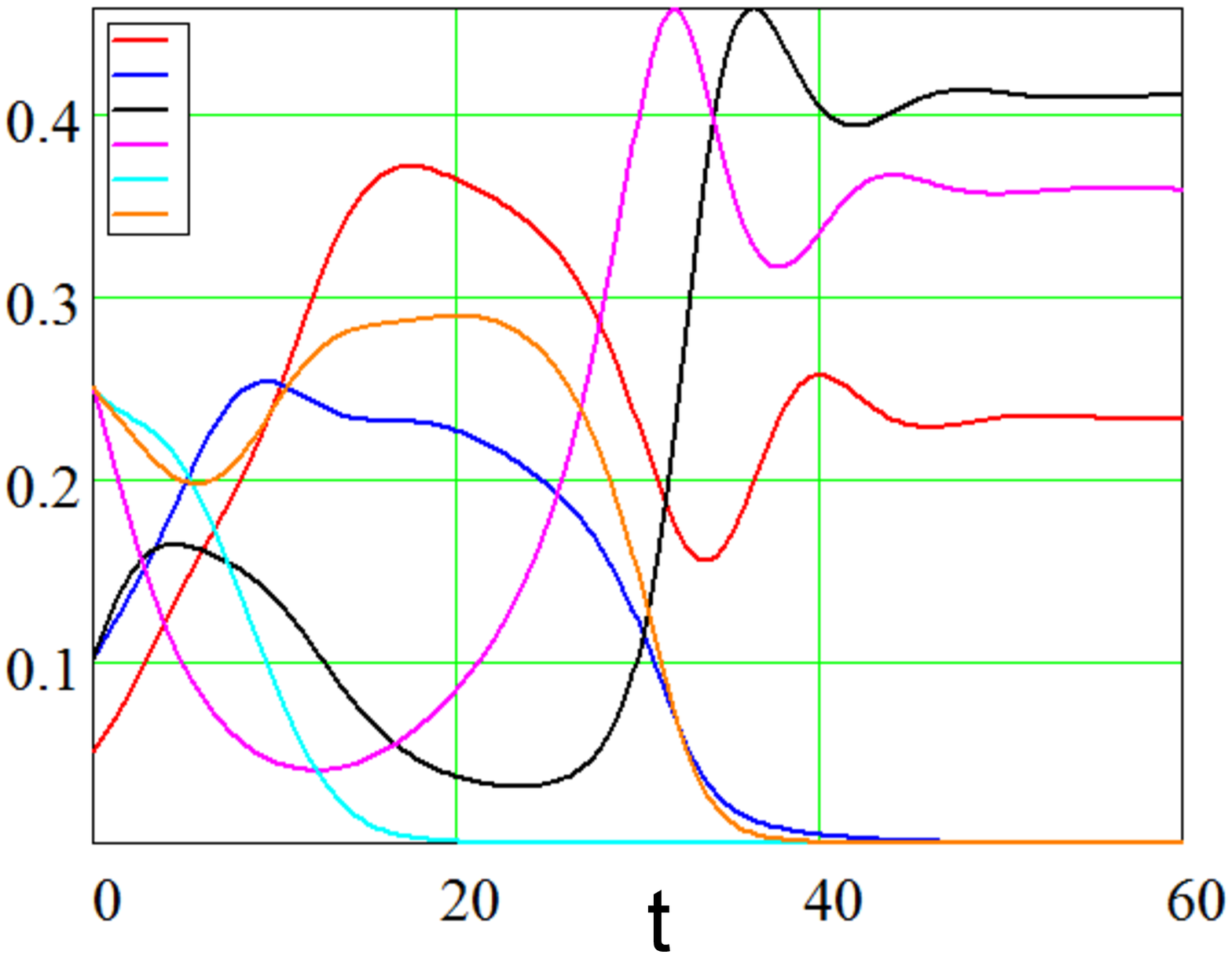}
\caption{Solutions to the replicator equation with the global regulation of the first kind \eqref{eq1:9} on $\Omega=(0,1)$ with interaction matrix \eqref{eq5:1} and with diffusion coefficients $\bs d_1=(0.4,0.5,0.4,0.5,0.4,0.5)$. Left panel: Solutions at the moment $t=60$. Right panel: The averages of the solutions $\overline{v}_i(t),\,i=1,\ldots,6$ depending on time $t$}\label{fig:3}
\end{figure}

For smaller diffusion coefficients $\bs d_2$ spatially nonhomogeneous stationary solutions appear (see Fig. \ref{fig:5}, left panel).\begin{figure}[!b]
\includegraphics[width=0.45\textwidth]{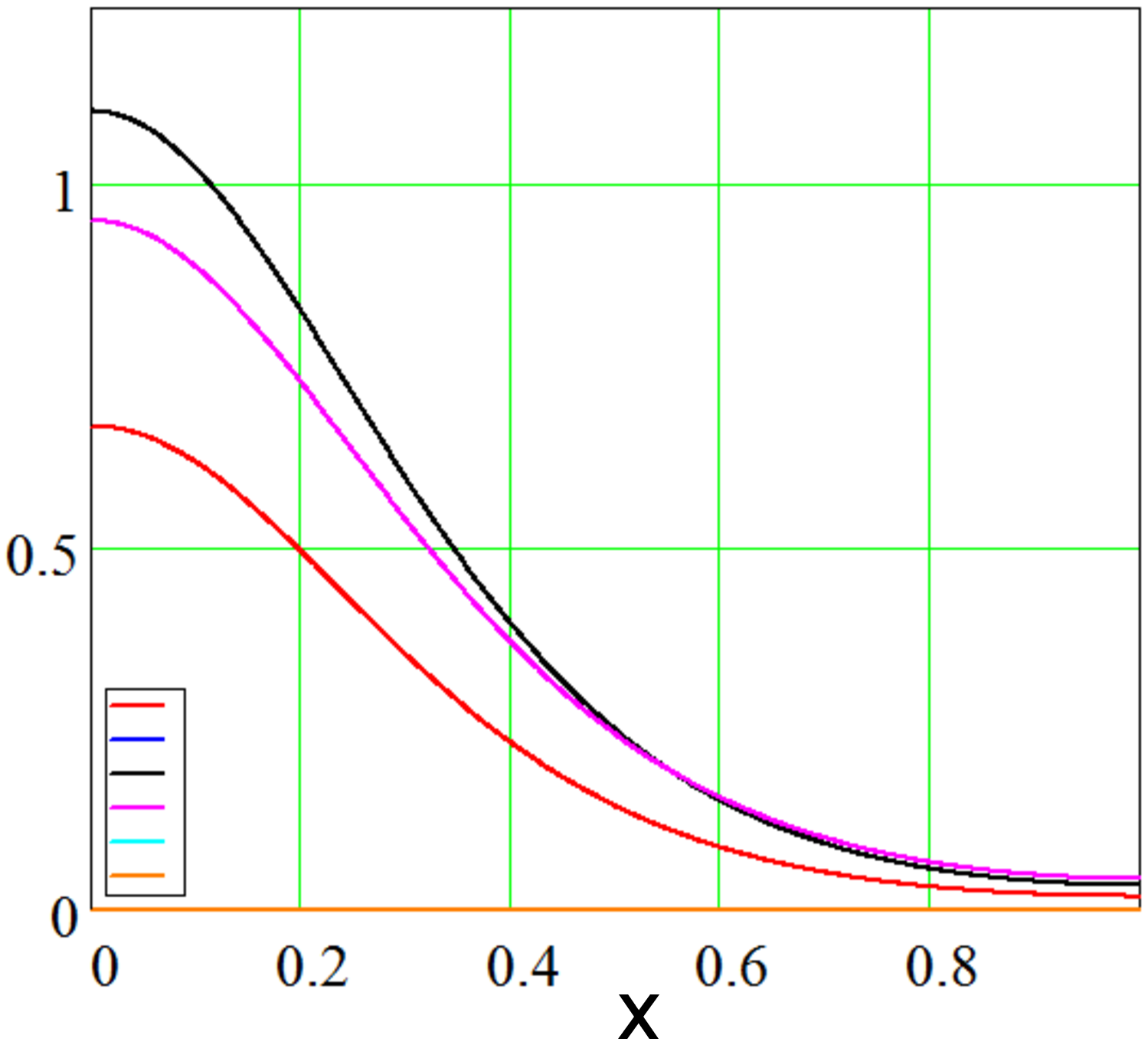}\hfill
\includegraphics[width=0.495\textwidth]{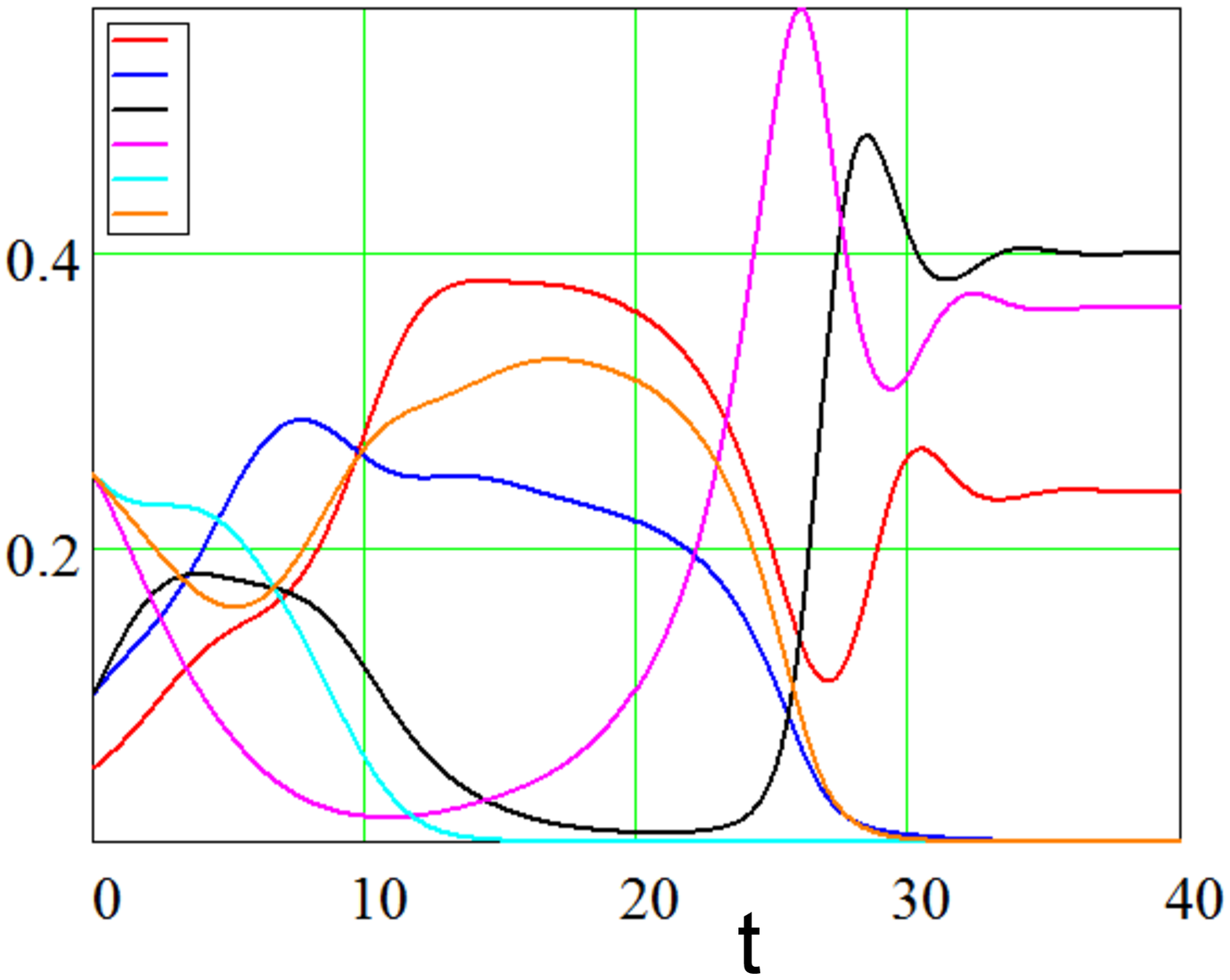}
\caption{Solutions to the replicator equation with the global regulation of the first kind \eqref{eq1:9} on $\Omega=(0,1)$ with interaction matrix \eqref{eq5:1} and with diffusion coefficients $\bs d_2=(0.04,0.05,0.04,0.05,0.04,0.05)$. Left panel: Solutions at the moment $t=40$. Right panel: The averages of the solutions $\overline{v}_i(t),\,i=1,\ldots,6$ depending on time $t$}\label{fig:5}
\end{figure} However, in the average, the behavior is still qualitatively similar to that of the homogeneous system: Three macromolecules persist whereas three others disappear from the system, which can be seen from the dynamics of the average values of the variables $\overline{v}_i(t)$ in the right panel of Fig. \ref{fig:5}.

As it was proved in \cite{bratus2011}, in the sense of the average behavior, we cannot expect qualitatively different behavior from the distributed replicator equation with the global regulation of the first kind.

\begin{figure}[!t]
\includegraphics[width=0.49\textwidth]{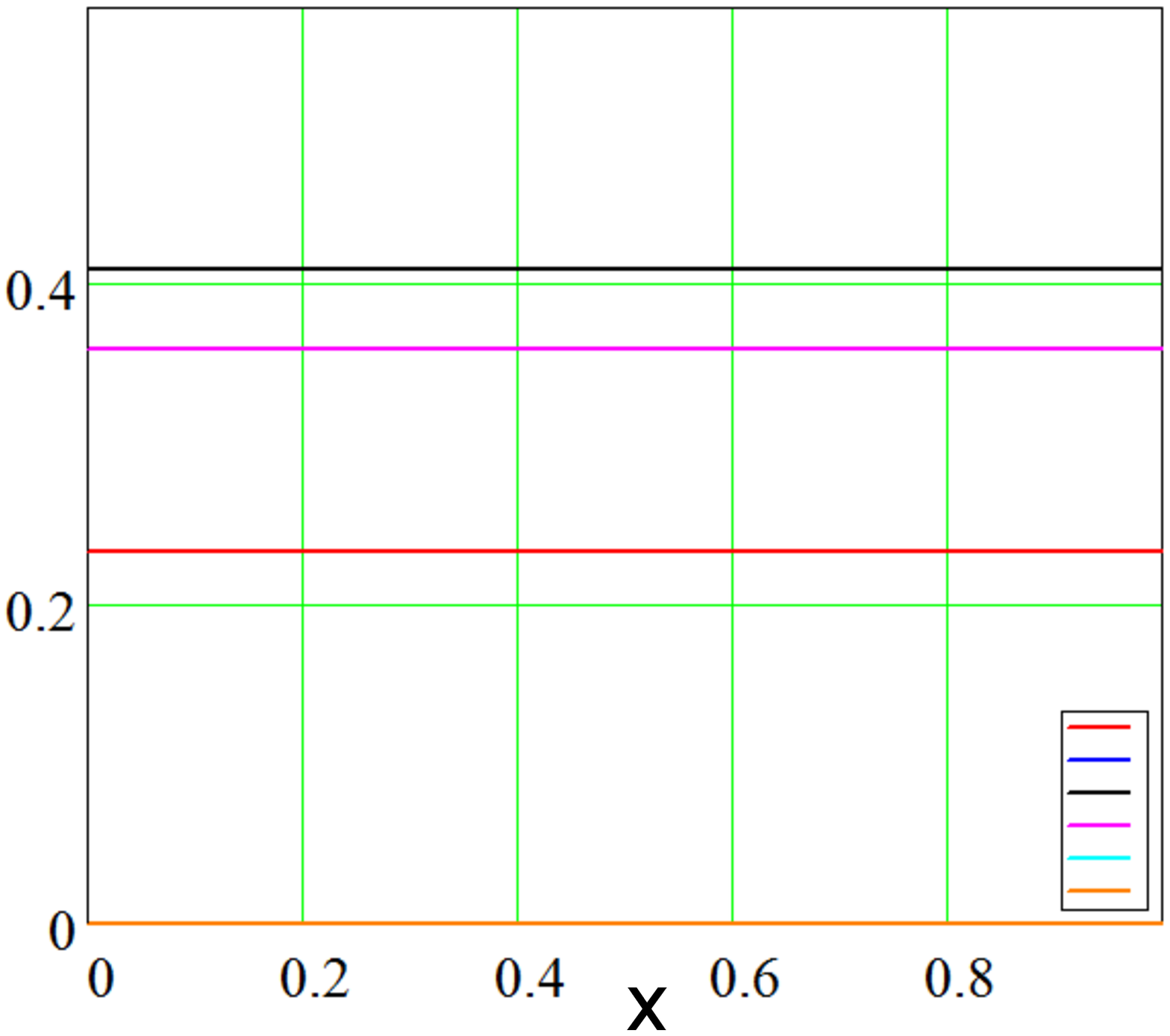}\hfill
\includegraphics[width=0.495\textwidth]{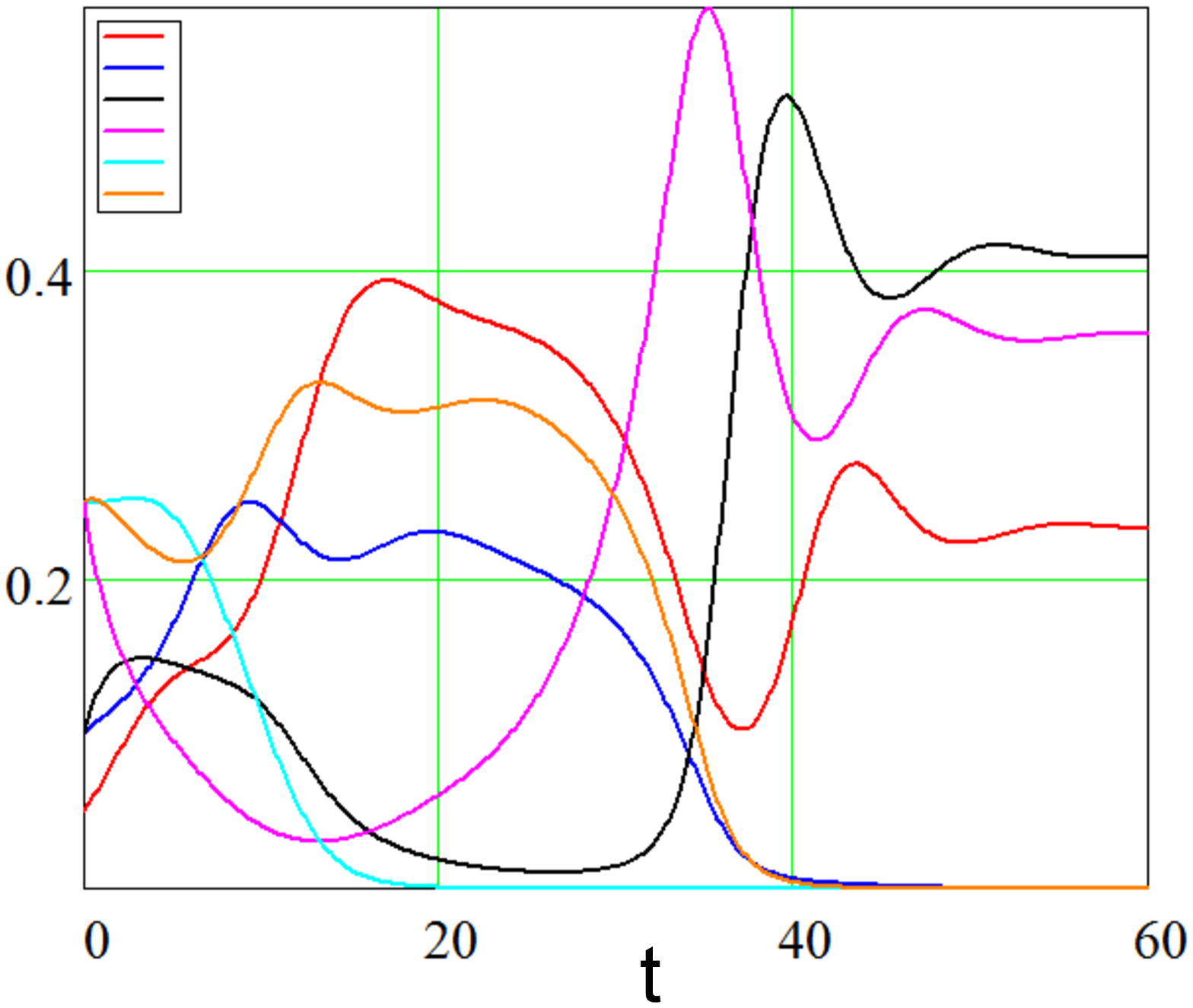}
\caption{Solutions to the replicator equation with the global regulation of the second kind \eqref{eq1:7} on $\Omega=(0,1)$ with interaction matrix \eqref{eq5:1} and with diffusion coefficients $\bs d_1=(0.4,0.5,0.4,0.5,0.4,0.5)$. Left panel: Solutions at the moment $t=60$. Right panel: The averages of the solutions $\overline{v}_i(t),\,i=1,\ldots,6$ depending on time $t$}\label{fig:6}
\end{figure}
A quite different picture is observed in the case of the reaction--diffusion replicator equation with the global regulation of the second kind \eqref{eq1:7}. In particular, while the diffusion coefficients are large enough, the solution behavior corresponds to the non-distributed case (as was proved in Sections 3 and 4). In Fig.~\ref{fig:6} it can be seen that, as well as in the previously discussed case of the global regulation of the first kind, there is an asymptotically stable spatially homogeneous stationary state, at which three macromolecules approach non-zero concentrations, whereas three others go extinct (cf. Figs. \ref{fig:3} and \ref{fig:5}). Decreasing the diffusion coefficients yields a qualitative change in the system behavior. Firstly, the solutions do not seem to approach a spatially uniform stationary state, the numerical calculations suggest that they keep oscillating. Secondly and most importantly, we observe that the concentrations of all six macromolecules are separates from zero, the system becomes permanent (see the right panel in Fig. \ref{fig:8}).
\begin{figure}[!ht]
\includegraphics[width=0.49\textwidth]{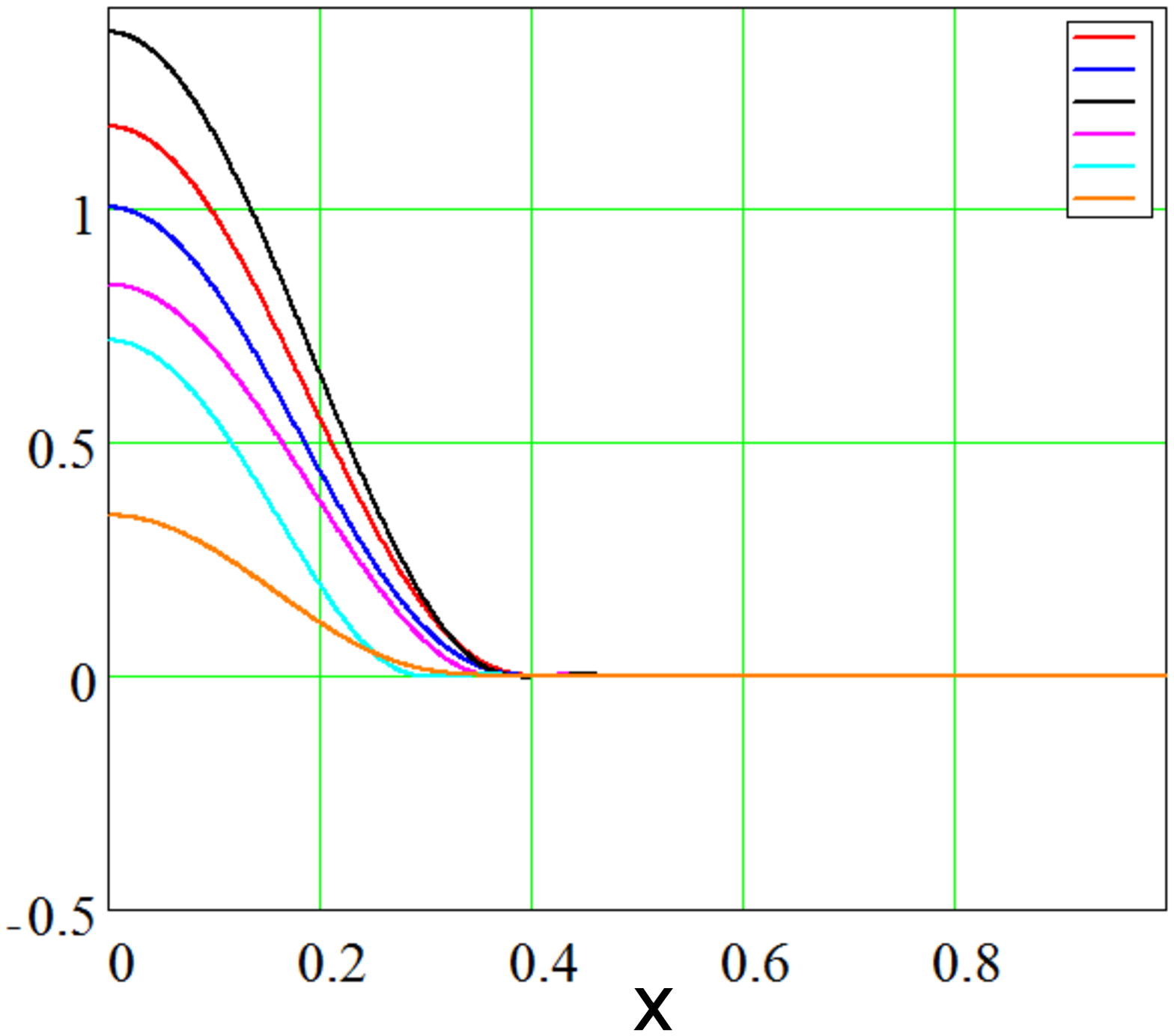}\hfill
\includegraphics[width=0.49\textwidth]{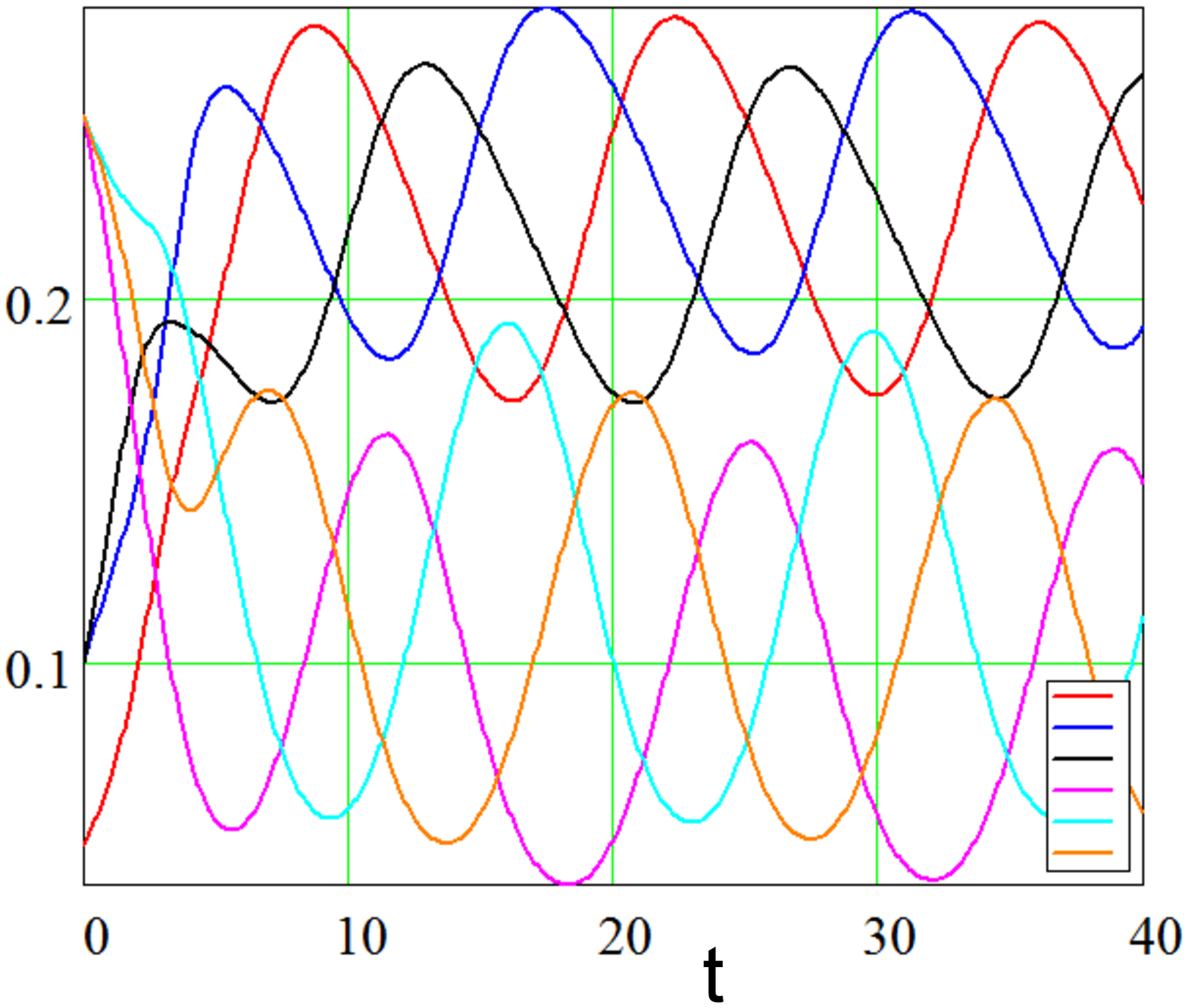}
\caption{Solutions to the replicator equation with the global regulation of the second kind \eqref{eq1:7} on $\Omega=(0,1)$ with interaction matrix \eqref{eq5:1} and with diffusion coefficients $\bs d_2=(0.04,0.05,0.04,0.05,0.04,0.05)$. Left panel: Solutions at the moment $t=40$. Right panel: The averages of the solutions $\overline{v}_i(t),\,i=1,\ldots,6$ depending on time $t$}\label{fig:8}
\end{figure}

In Fig. \ref{fig:9} time dependent solutions in the space $(x,t)$ are shown that correspond to the case of Fig. \ref{fig:8}.
\begin{figure}[!b]
\includegraphics[width=0.33\textwidth]{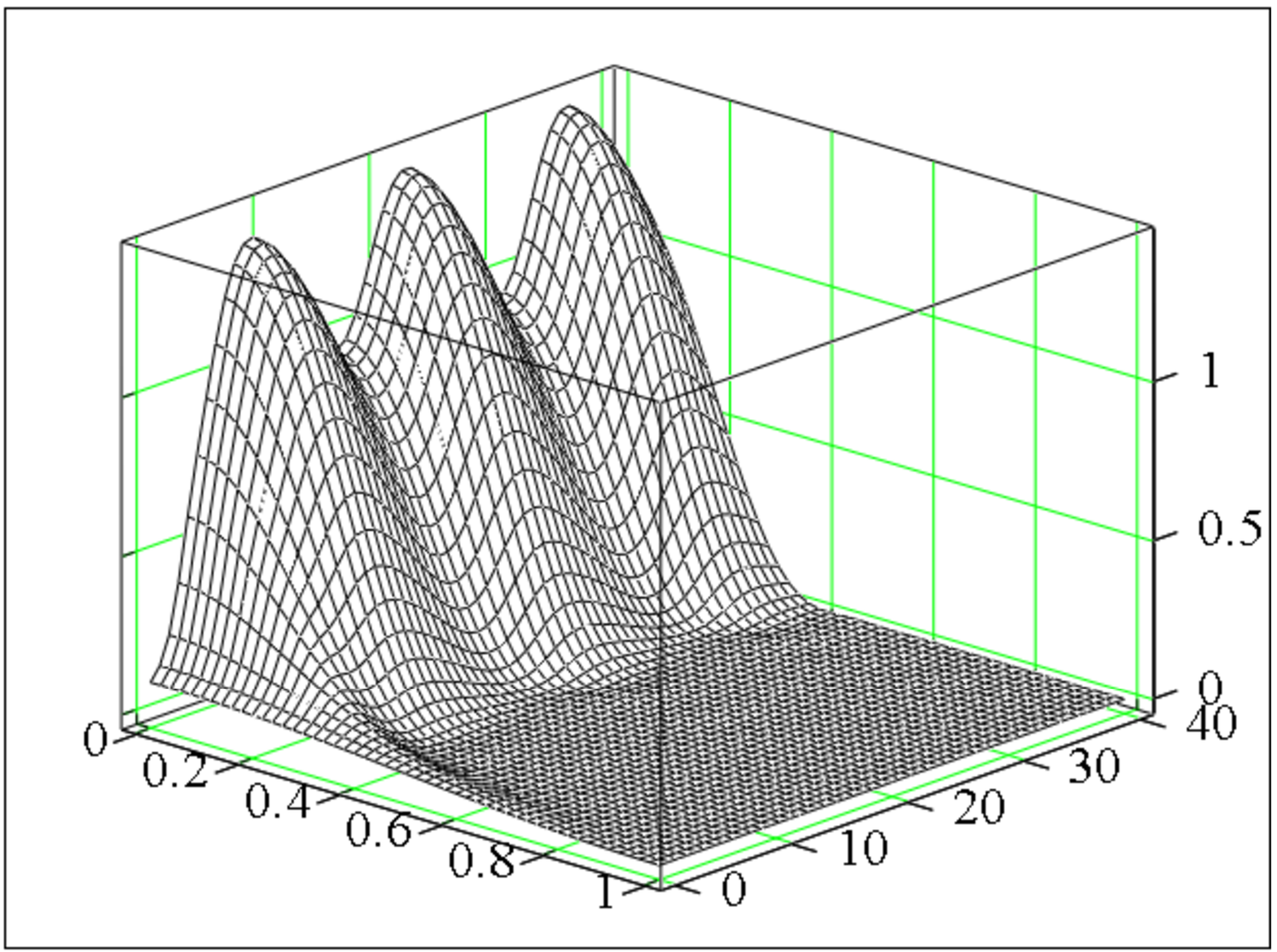}
\includegraphics[width=0.33\textwidth]{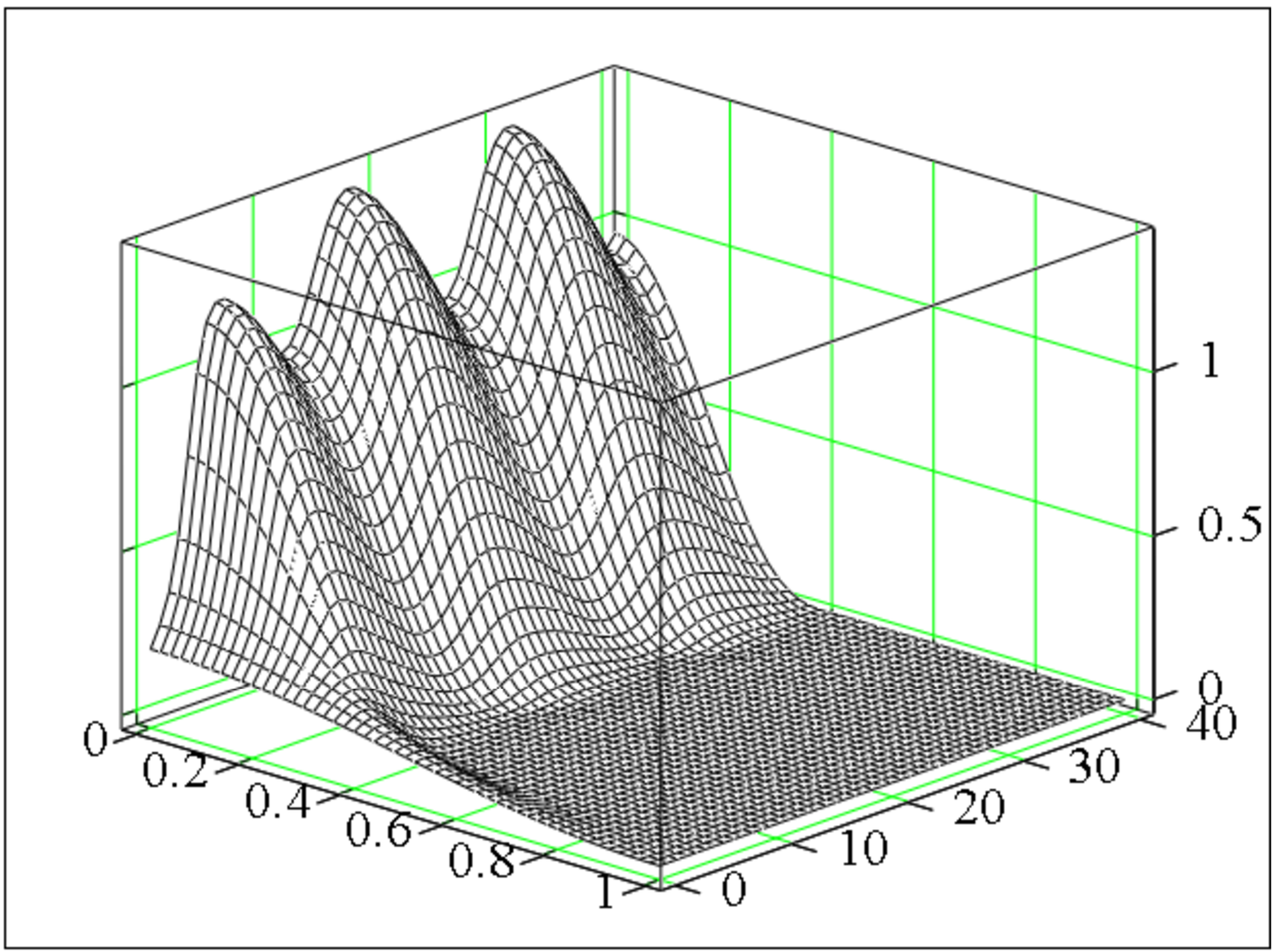}
\includegraphics[width=0.33\textwidth]{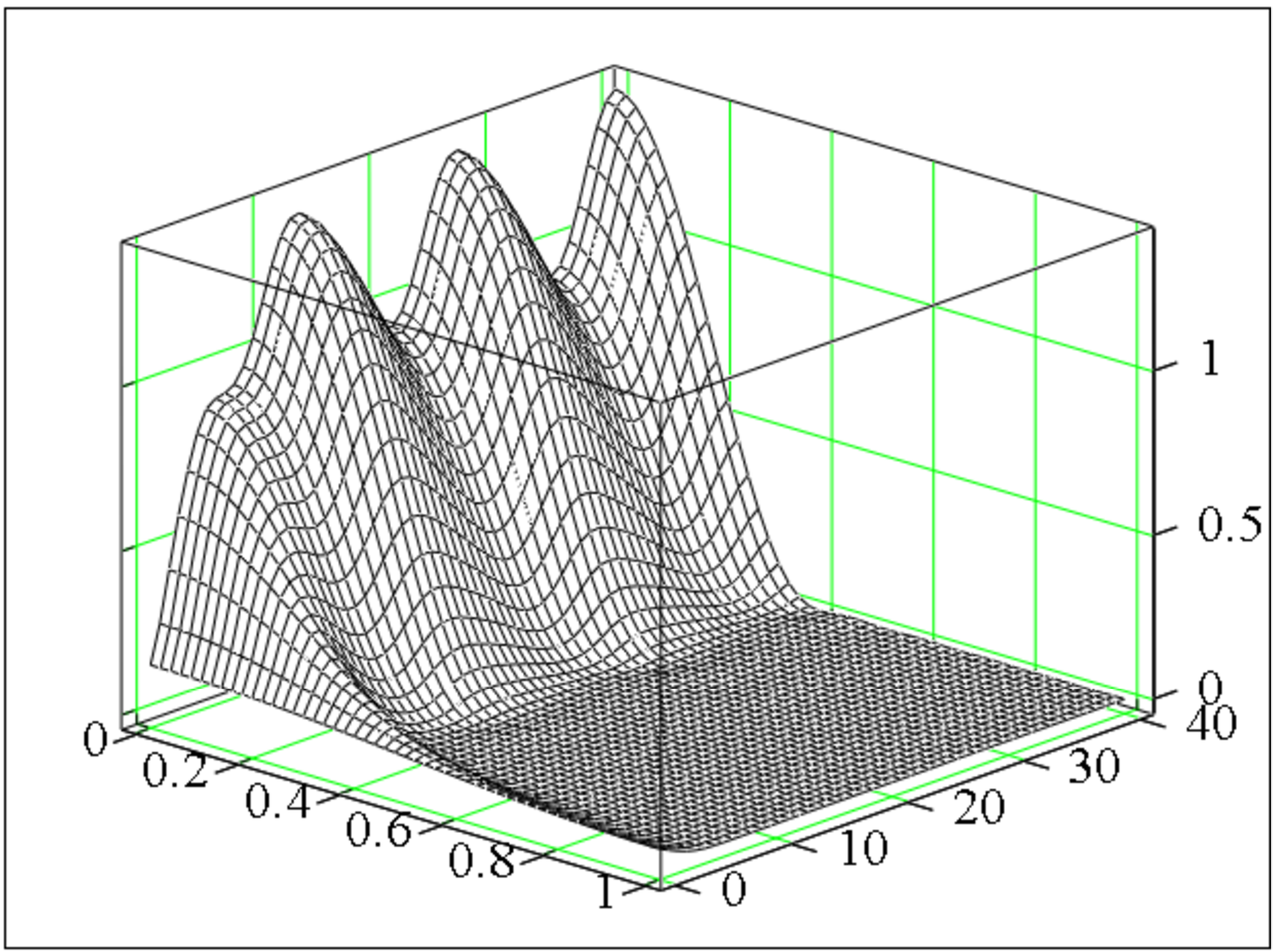}
\includegraphics[width=0.33\textwidth]{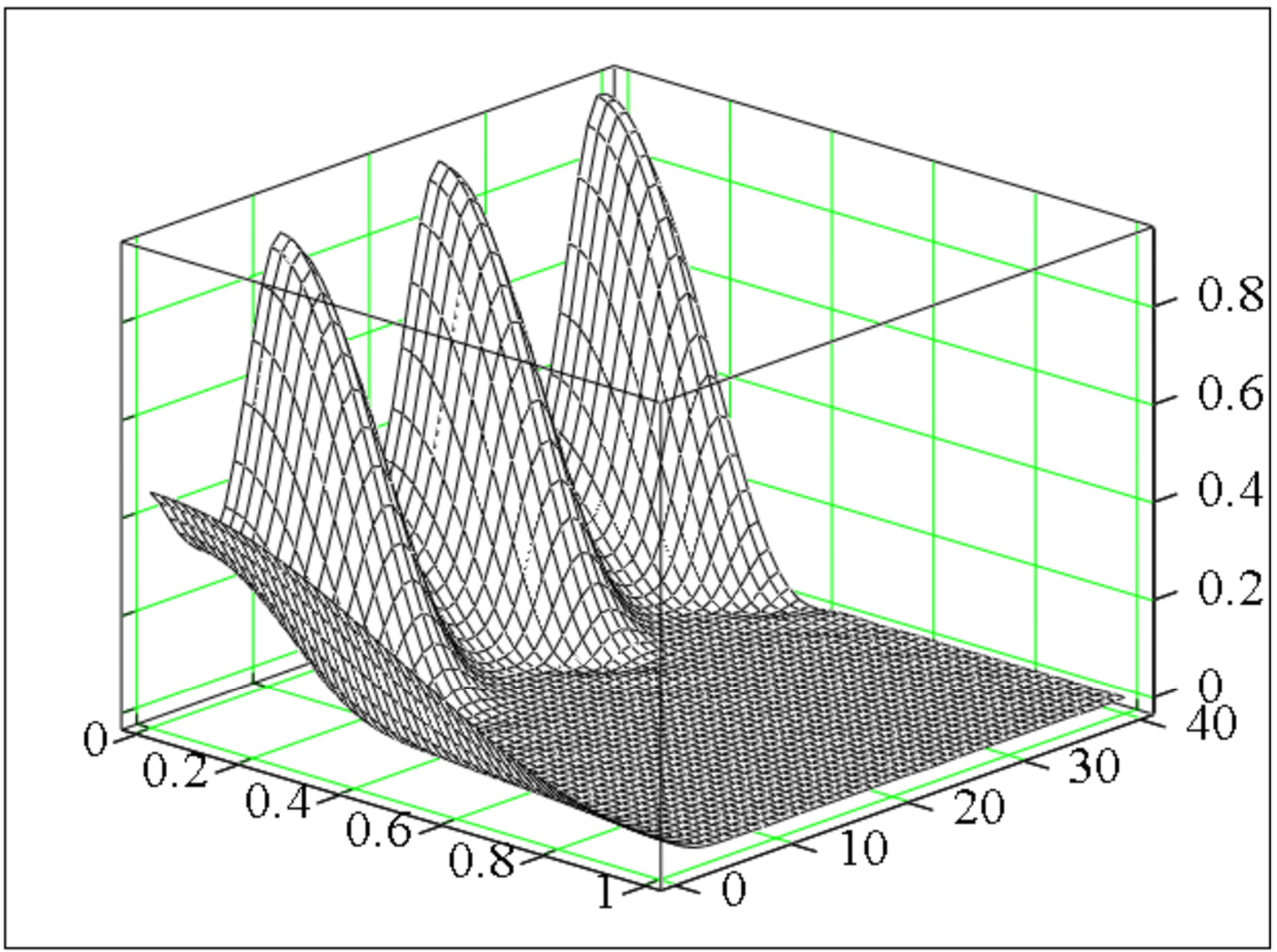}
\includegraphics[width=0.33\textwidth]{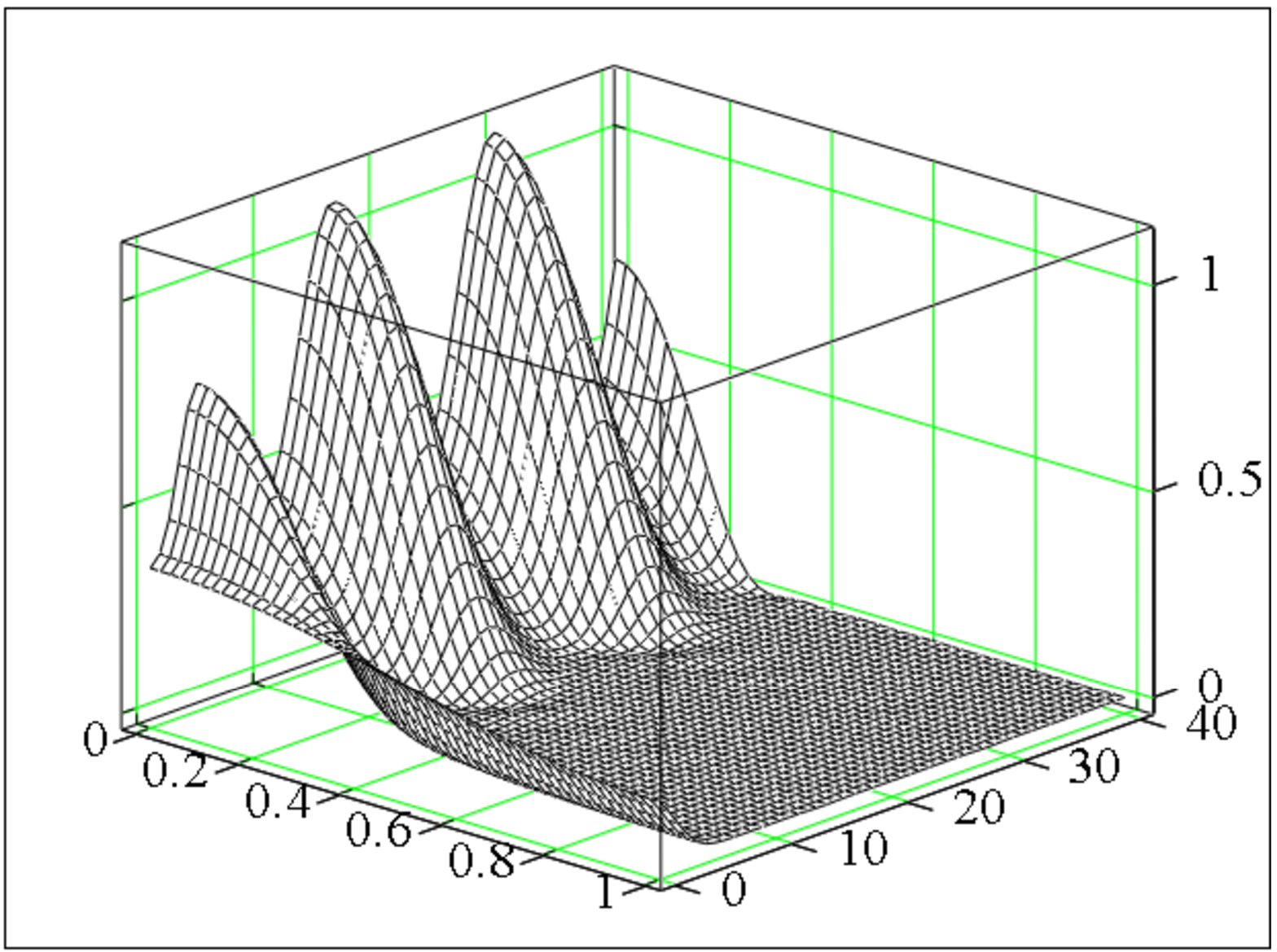}
\includegraphics[width=0.33\textwidth]{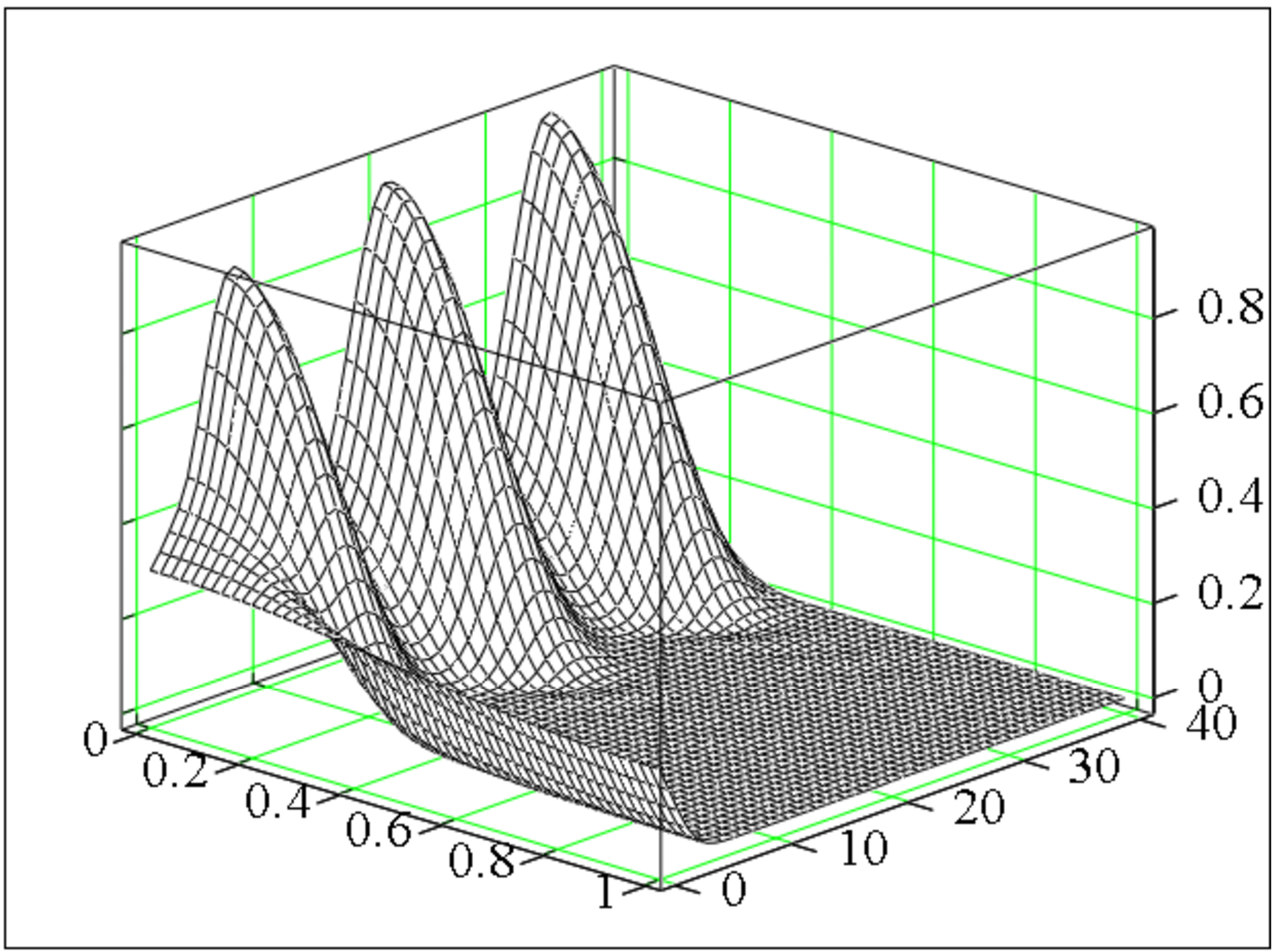}
\caption{Solutions in the $(x,t)$ space to the problem \eqref{eq1:7} on $\Omega=(0,1)$ with interaction matrix \eqref{eq5:1} and with diffusion coefficients $\bs d_2$. The averaged values of the variables $\overline{v}_i(t)$ are shown in the right panel in Fig. \ref{fig:8} }\label{fig:9}
\end{figure}

\clearpage
\section{Concluding remarks}
We presented in this manuscript an analytical and numerical analysis of a parabolic type equation, which we call the replicator equation with the global regulation of the second kind. We showed that this equation possesses interesting features that allow considering it as a viable candidate for the replicator equation with explicit spatial structure. In particular, our analytical and numerical analysis shows that
\begin{itemize}
\item For sufficiently large diffusion coefficients it is reasonable to expect that the behavior of the solutions to the replicator equation with the global regulation of the second kind can be inferred from the analysis of the solutions of the corresponding non-distributed replicator equation (Sections 3 and 4). As expected, in this case the mean-field approximation of a well-stirred reactor works well.
\item Some of the results concerning the permanence of the classical replicator equation can be used to obtain sufficient conditions for the permanence or persistence of the replicator equation with the global regulation of the second kind (Section 5)
\item Most importantly, we were able to show numerically that for sufficiently small diffusion coefficients it can be expected that the properties of the distributed system differ significantly from the properties of the non-distributed system. In particular, the global regulation of the second kind mediates coexistence of different macromolecules. The importance of our results follows also from the fact that the particular replicator system we consider in Section 6 is based on the in vitro experiments, in which it was shown that all six macromolecules survive.
\end{itemize}

In conclusion we note that the important phenomena observed numerically in Section 6 call for analytical proofs, and this is one of our ongoing projects to supplement the numerical findings of the current text with analytical theory.

\paragraph{Acknowledgements:} This research is supported in part by the Russian Foundation for Basic Research (RFBR) grant \#10-01-00374 and joint
grant between RFBR and Taiwan National Council \#12-01-92004HHC-a. ASN's research is supported in part by ND EPSCoR and NSF grant \#EPS-0814442.


\end{document}